\newif\ifEprint
\Eprinttrue

\newif\ifDynamic
\Dynamicfalse

\documentclass[11pt,english]{article}

\usepackage[T1]{fontenc}
\usepackage[showframe=false,margin=1in]{geometry}
\usepackage{changepage,float}
\usepackage{babel,amsmath, amssymb, mathtext, algorithm2e,
  algorithmic, textcomp, subfig, graphicx, booktabs, multirow, hhline,
  xcolor, bm, afterpage}
  \usepackage{framed}

\usepackage[compact,small]{titlesec}

\usepackage{enumitem}

\usepackage{tabularx}
\usepackage{array}

\title{Verifiable Member and Order Queries on a List in Zero-Knowledge}

\author{Esha Ghosh\thanks{Brown University} \and Olga Ohrimenko\thanks{Microsoft Research} \and Roberto Tamassia\footnotemark[1]}
\date{}

\newcommand{\List}{\mathcal{L}}

\newcommand{\Q}{\mathsf{Query}}
\newcommand{\R}{\Omega_\llist}
\newcommand{\Rgen}{\Omega}

\newcommand{\Sig}{\Sigma_\mathcal{L}}

\newcommand{\ZZ}{\mathbb{Z}}
\newcommand{\advM}{\mathcal{A}}

\newcommand{\authunitm}[1]{$t_{{#1}\in{\mathcal{L}}}$}
\newcommand{\authunitoM}[2]{t_{{#1}<{#2}}}
\newcommand{\authunito}[2]{$t_{{#1}<{#2}}$}

\newcommand{\elem}{\mathsf{Elements}}
\newcommand{\ele}{$\elem$}

\newcommand{\li}{$\mathcal{L}$}
\newcommand{\llist}{\mathcal{L}}

\newcommand{\permM}[2]{{\pi}_{#1}({#2})}
\newcommand{\perm}[2]{$\permM{#1}{#2}$}

\newcommand{\ppadsl}{$\mathsf{PPAL}$}
\newcommand{\ppads}{\mathsf{PPAL}}
\newcommand{\p}{\mathsf{poly(k)}}
\newcommand{\query}{\mathsf{\delta}}

\newcommand{\salt}{\mathsf{salt}}
\newcommand{\setupM}{\mathsf{Setup}}
\newcommand{\setup}{$\mathsf{Setup}$}

\newcommand{\sublist}{\delta}

\newcommand{\GG}{\mathcal{G}}

\newcommand{\D}{\mathsf{digest}_C}
\newcommand{\K}{\mathsf{digest}_S}
\newcommand{\Or}{\mathsf{order}}
\newcommand{\pr}{\mathsf{proof}}
\newcommand{\ve}{\mathsf{Verify}}
\newcommand{\Sim}{\mathsf{Sim}}
\newcommand{\advN}{\mathsf{Adv}}
\newcommand{\st}{\mathsf{state}}

\newcommand{\ZS}{\mathsf{Setup}}
\newcommand{\ZP}{\mathsf{Prover}}
\newcommand{\ZV}{\mathsf{Verifier}}
\newcommand{\f}{\mathsf{flag}}
\newcommand{\PK}{\mathsf{PK}}

\newcommand{\com}{\mathsf{com}}
\newcommand{\mem}{\mathsf{member}}

\newcommand{\TK}{\mathsf{TK}}

\newcommand{\HC}{\mathsf{HardComm}}
\newcommand{\SC}{\mathsf{SoftComm}}
\newcommand{\MS}{\mathsf{MercSetup}}
\newcommand{\open}{\mathsf{MercOpen}}
\newcommand{\tease}{\mathsf{Tease}}
\newcommand{\VO}{\mathsf{VerOpen}}
\newcommand{\VT}{\mathsf{VerTease}}
\newcommand{\PD}{\ZP_D}
\newcommand{\VD}{\ZV_D}
\newcommand{\n}{\mathsf{nil}}

\newcommand{\cS}{\mathsf{IntComSetup}}
\newcommand{\IC}{\mathsf{IntCom}}
\newcommand{\IO}{\mathsf{IntComOpen}}
\newcommand{\HomIntCom}{\mathsf{HomIntCom}}

\newcommand{\rank}{\mathsf{rank}}
\newcommand{\Ev}{\mathbb{E}}

\newcommand{\sICS}{\mathsf{SimICSetup}}

\newcommand{\sZS}{\mathsf{SimZKSSetup}}

\newcommand{\Pro}{\mathsf{Protocol}}

\newcommand{\zks}{\mathsf{ZKS}}
\newcommand{\Pone}{\mathsf{P}_1}
\newcommand{\Ptwo}{\mathsf{P}_2}

\newcommand{\x}{v}

\newtheorem{Definition}{Definition}[section]

\newtheorem{Theorem}{Theorem}[section]

\newenvironment{proof}[1][Proof]{\begin{trivlist}
\item[\hskip \labelsep {\bfseries #1}]}{\end{trivlist}}

\newcommand{\qed}{\nobreak \ifvmode \relax \else
      \ifdim\lastskip<1.5em \hskip-\lastskip
      \hskip1.5em plus0em minus0.5em \fi \nobreak
      \vrule height0.75em width0.5em depth0.25em\fi}
\newcommand\func[1][]{\mathit{f}\ifx\\#1\\\else(#1)\fi}

\begin{document}

\maketitle

\begin{abstract}

We introduce a formal model for order queries on lists in zero knowledge in the traditional authenticated data structure model.
We call this model Privacy-Preserving Authenticated List (PPAL).
In this model, the queries are performed on the list stored in the (untrusted) cloud where data integrity and privacy have to
be maintained. To realize an efficient authenticated data structure, we first adapt consistent data query model.
To this end we introduce a formal model called Zero-Knowledge List (ZKL) scheme which generalizes consistent membership queries in zero-knowledge 
to consistent membership and order queries on a totally ordered set in zero knowledge. We present a construction of ZKL based on zero-knowledge set
and homomorphic integer commitment scheme. Then we discuss why this construction is not as efficient as desired in cloud applications and
present an efficient construction of PPAL based on bilinear accumulators and bilinear maps which is provably secure and zero-knowledge.

\end{abstract}

\section{Introduction}
\label{sec:intro}

\ifEprint

Authentically releasing partial information while maintaining privacy is a well known requirement in many practical scenarios where the data being dealt with is sensitive. 

Consider, for example, 
the following medical case study presented in \cite{MedHR}.
Each patient has a personal health record (PHR) that contains the medication and vaccination history of the patient. Entries are made against the dates when medicines are taken and
vaccinations are performed.
Thus, the PHR is a chronologically sorted document signed by the medical provider and given
to the patient.
Now the patient might need to authorize the release of a subset of the PHR with only the relevant information to be sent to third parties on an as needed basis, without the medical provider's involvement.
For example, let us say the patient wants to join a summer camp that requires the vaccination record of the patient and the order in which the vaccinations were taken.
the patient wants to release the relevant information in a way such that the camp can verify that the data came from a legitimate medical provider, but 
the camp cannot learn anything beyond the authorized subset of relevant information, i.e., the vaccinations and their chronological order, but not the exact date when they were taken.

Consider another example where there are multiple regional sales divisions of a company distributed across three neighboring states.
A monthly sales report contains the number of products sold by each of the divisions, arranged in non-decreasing order. Each monthly sales report is signed by the authority and stored on a cloud server.
By the company's access control policy, each sales division is allowed to learn how it did in comparison to the other units, but not anything else. That is, a division cannot learn the sales numbers of other divisions or their relative performance beyond what it can infer by the comparisons with itself. 
Thus, the cloud would need to release the relevant information in such a way that the querying division can verify the data came from the legitimate source but not learn anything beyond the query result. 

The above examples motivate the following model: a trusted \emph{owner} generates an ordered set of elements. Let us call this ordered set a \emph{list}.
The owner outsources the list to a (possibly untrusted) party, let us call it \emph{server}. There is another party involved who issues order queries on the list, let us call this party \emph{client}.
The client only interacts with the server. So the server has to release information in a way such that the client can verify the authenticity of the data it receives, i.e., that it is truly generated by
the trusted owner. But the client should not be able to learn anything beyond the answers to its queries. 

This above model specifies an \emph{authenticated data structure} with an additional  privacy requirement.
An \emph{authenticated data structure} \cite{Tamassia03} is a structured collection of data (e.g., a list, tree, graph, or map) along with a set of query operations defined on it.
Three parties are involved in an authenticated data structure (ADS) scheme, namely, the data owner, the server and the client/user. ADS framework
allows the data owner to outsource data processing tasks to an untrusted
server without loss of data integrity for clients. This is achieved as follows.
The (trusted) data owner produces authentication information about the dataset (ordered list in our case) and a short digest signature and sends a copy of the dataset along with the authentication information to
the (untrusted) server and the digest signature to the client. The server responds to the (legitimate) client queries about the dataset by returning
the query answer and a compact proof of the answer. The client uses
the digest signature (obtained from the owner), the query answer and the proof obtained from the server to verify the integrity of the answer.

Classic hash-based authenticated data structures were designed without taking into account privacy goals and provide proofs that leak information about the dataset beyond the query answer. 
For example, in a hash tree \cite{Merkle80,Merkle89} for a set of $n$ elements, the proof of the membership of an element in the set has size $\log n$, thus leaking information about the size
of the set. Also, if the elements are stored at the leaves in sorted order, 
the proof of \emph{membership} of an element reveals its rank. Similar information leaks occur in other hash-based authenticated data structures for 
dictionaries and maps, such as authenticated skips lists~\cite{Goodrich01}.
As another example, consider an approach for supporting \emph{non-membership} proofs
using an authenticated data structure that supports membership proofs.
This method involves storing intervals of consecutive elements $(x_i, x_{i+1})$
and returning as a proof of non-membership of a query 
element $x$ the interval $(x_i, x_{i+1})$ such that $x_i < x < x_{i+1}$.
Hence, the proof trivially reveals two elements of the set. 

We define a \emph{privacy-preserving authenticated data structure} as
an ADS with an additional privacy property that  ensures that the proof returned by the server to the client does not reveal any information about 
the dataset beyond what can be learned from the current and previous answers to queries to the dataset.
In this paper we study one such data structure, a \emph{privacy-preserving authenticated list} (PPAL),
for which we consider order queries.

A privacy-preserving authenticated list (PPAL) allows an owner to outsource to the server data with different access control policies
imposed on it. Hence, when the owner outsources it to the server, clients can access only parts of it from the server and verify that it is indeed owner's
authentic data but should not be allowed to learn about the data they do not have permission to access.
Hence, privacy policy should be also imposed on the proofs of authenticity of the data that the clients learn
(the property not supported by classical ADS). PPAL has several interesting applications as we have already seen in the motivating examples.
We also envision a PPAL list to be an important building block for designing efficient hierarchical privacy-preserving data structures
e.g., ordered trees that store XML data.

In this paper, we present an efficient solution for privacy-preserving authenticated lists that supports
queries on the relative order of two (or more) elements of the list.
This framework guarantees integrity of the order queries through a compact proof returned to the client.
The proof does not reveal the actual ranks of the elements nor
any order information between elements other that what can be inferred from the current and previous answers by the rule of transitivity.

We first present a generic approach to this problem in the traditional consistent query model \cite{MicaliRK03,ChaseHLMR05,Ostrovsky04,Catalano:2008,Libert2010} where there are two parties involved: the prover, $\ZP=(\Pone,\Ptwo)$ 
and the verifier, $\ZV$. $\Pone$ takes an ordered list as input and produces a short commitment which is made public. Then $\ZV$ generates membership and order 
queries on the list and $\Ptwo$ responds with the answers and the proofs. Once $\ZP$ commits to a list, it cannot give answers inconsistent with the commitment that will pass the
verification test by $\ZV$. We formalize this framework as \emph{Zero Knowledge List (ZKL)} and give a construction in Section~\ref{sec:ZKL}.

It is easy to see this model can be interpreted in the PPAL framework as follows. We can make the owner run $\Pone$, the server run $\Ptwo$ and the client run
$\ZV$. In fact in the ZKL model we get stronger security guarantee as once the owner commits to a list, even a malicious one, 
it cannot give inconsistent answers later. Moreover the ZKL framework supports both membership and order queries in Zero Knowledge.
However, as we discuss in Section~\ref{sssec:ZKLEfficiency}, the construction is not very practical for cloud computing scenario where the client can be a mobile device.

The ZKL model indeed gives stronger security guarantees, but in the privacy-preserving authenticated list (PPAL) model, the owner is in fact a \emph{trusted} party.
We discuss the PPAL framework in Section~\ref{sec:scheme} and design an efficient PPAL scheme, exploiting the fact that the owner is an honest party. In our scheme, for a source list of size $n$ and a query of size $m$,
neither the server nor the owner runs in time more than $O(n)$ or requires storage more that $O(n)$. The running time for the server can be brought down to $O(m\log n)$ with $O(n)$
preprocessing time and the client requires time and space proportional to $O(m)$. We give the construction and discuss its efficiency in Section~\ref{sec:construction}.

\fi


\subsection{Problem Statement}
\label{sec:problem-statement}
In this section we define order queries, introduce our adversarial model and security properties, and discuss our efficiency goals.

\subsubsection{Query} 
\label{ss:querystructure}

Let $\List$ be a linearly ordered list of non-repeated elements.
An \emph{order query} on a list $\llist$ of distinct elements is defined as follows: 
given a pair of query elements $(x,y)$ of~$\List$, the server returns the pair with its elements 
rearranged according to their order in~$\List$ together with proofs of membership of $x$ and $y$
and a proof of the returned order. For example, if $y$ precedes $x$ in $\List$, then the pair $(y,x)$ is returned as an answer.

For generality, the data structure also supports \emph{batch order query}:
Given a list of query elements $\sublist$, the server returns a permutation of $\sublist$ according to the ordering of the elements in~$\List$, 
together with a proof of the membership of the elements and of their ordering in~$\List$.

The above model captures the query model of a privacy-preserving authenticated list.
In comparison, zero knowledge list structure supports the same queries as well as
non-membership queries. Hence, as a response to a (non-)membership
element query the prover returns a boolean value indicating if the element is in the list and a corresponding proof of (non-)membership.
              
\subsubsection{Adversarial model and security properties} 
In this section we present adversarial models and security properties of PPAL and ZKL.

Following the authenticated data structure model, list $\List$ plus authentication information about it is created by a data owner and given to a server, who answers queries on $\List$ issued by a client, who verifies the answers and proofs returned by the server using the public key of the data owner.
We assume the data owner is trusted by the client, who has the public key of the data owner. However, both the client and server can act as adversaries, as follows:
\begin{itemize}[noitemsep,nolistsep]
\item The server is malicious and may try to forge proofs for incorrect answers to (ordering) queries. 
For example, the server may try to prove an incorrect ordering of a pair of elements of~$\List$. 
 
\item The client tries to learn from the proofs additional information about list $\List$ beyond what it has inferred from the answers. 
For example, if the client has performed two ordering queries with answers $x<y$ and $x<z$, it may want to find out whether $y<z$ or~$z<y$.

 \end{itemize}
Note that in typical cloud database applications, the client is allowed to have only a restricted view of the data structure and 
the server enforces an access control policy that prevents the client from 
getting answers to unauthorized queries.  This motivates the curious behavior by the client.
The client may behave maliciously and try to ask ill-formed queries or queries violating the access control policy. But the server may just refuse to answer when the client asks illegal queries. 
So the client's legitimate behavior can be enforced by the server.

We wish to construct a privacy-preserving authenticated data structure for list $\List$, i.e., a data structure with the following security properties:

\begin{description} [noitemsep,nolistsep]
 \item [Completeness] ensures that honestly generated proofs are always accepted by the client.
 \item [Soundness] mandates that proofs forged by the server for incorrect answers to queries do not pass the verification performed by the client.
 \item [Zero Knowledge] means that each proof received by the client to a query reveals and verifies the answer and nothing else.
In other words, for any element $x_i \in \{0,1\}^\ast$, the simulator, given oracle access to $\llist$, should be able to simulate proofs for order queries that are indistinguishable from real proofs.
\end{description}
To understand the strength of the zero-knowledge property, let us illustrate to what extent the proofs are non-revealing.
One of the guarantees of this property is that receiving a response to a query~$\sublist$
does not reveal where in~$\llist$ queried elements of~$\sublist$ are.
In other words, no information about $\llist$, other than
what is queried for in~$\sublist$ is revealed. It is worth noting that in the context of leakage-free redactable signature schemes, this property has been
referred to as \emph{transparency}~\cite{Brzuska10,Samelin12} and \emph{privacy}~\cite{Chang09,Kundu12}.
Moreover, zero knowledge also provides security for the size of the list $\llist$ from the client.

Since we let the client ask multiple queries on a static list adaptively,
in principle, it is possible that even though the individual query responses and proofs
do not leak any extraneous information about the source list, when the responses and proofs are collected together, the client is able to infer some structural information about the source
list, which it had not explicitly queried for.
Hence, we need to ensure that the scheme is immune against any potential leakage
of any structural information
that has not been explicitly asked for by the client.
More concretely,
in a linearly ordered list~$\llist$,
the client should not be able to infer any relative order that is not inferable by the rule of transitivity from the queried orders.
This security guarantee also follows from the zero-knowledge property.

The adversarial model in ZKL is different from that of PPAL.
The ZKL model considers only two parties: the prover and the verifier.
The prover initially computes a commitment to a list~$\llist$ and makes this commitment public (i.e., the verifier also receives it).
Later the verifier asks membership and order queries on the list and the prover responds accordingly. In ZKL both the prover and the verifier can be malicious
as follows:
\begin{itemize}[noitemsep,nolistsep]
 \item The prover may try to give answers which are inconsistent with the initial commitment.
 \item The adversarial behavior of the verifier is the same as that of the client in the PPAL model.
\end{itemize}
The security properties of ZKL (Completeness, Soundness, Zero-Knowledge) guarantee security against malicious prover and verifier.
They are close to the ones of PPAL except for Soundness which captures that the prover can try to create a forgery on a list of his choice.
We discuss the security properties of ZKL in more detail in Section~\ref{ssec:ZKLSec}.

\subsubsection{Efficiency} 

We characterize the efficiency of a privacy-preserving authenticated data structure for a list, $\List$, of~$n$~items by means of the following space and time complexity measures:

\begin{itemize}[noitemsep,nolistsep]
\item \emph{Server storage:} Space at the server for storing list $\List$ and the authentication information and for processing queries. 
Ideally, the server storage is~$O(n)$, irrespective of the number of queries answered.
\item \emph{Proof size:} Size of the proof returned by the server to the client. Ideally, the proof has size proportional to the answer size.
\item \emph{Setup time:} Work performed by the data owner to create the authentication information that is sent to the server. Ideally, this should be~$O(n)$.
\item \emph{Query time:} Work performed by the server to answer a query and produce its proof. Ideally, this work is proportional to the answer size.
\item \emph{Verification time:} Work performed by the client to verify the answer to a query using the proof provided by the server and the public key of the data owner. Ideally, this work is proportional to the answer size.
\end{itemize}

\subsection{Related Work}
\label{sec:related}

We discuss related literature in three sections. First, we discuss
work on data structures that answer queries in zero knowledge.
This work is the closest to our work on zero knowledge lists.
We then discuss signature schemes that can be interpreted in the privacy-preserving authenticated data structure model.
Finally, we highlight the body of literature regarding leakage-free redactable signature schemes for ordered lists in detail.
The latter is the closest to the problem of privacy preserving authenticated lists that we are addressing in this manuscript. 

\paragraph{Zero Knowledge Data Structures}

Buldas~\textit{et al.}~\cite{Buldas:2002}
showed how to prove answers to dictionary queries using
an authenticated search tree-based construction, but did not consider privacy.
For a set of size $n$, the construction produces a proof of \mbox{(non)m}embership
of an element in the set that has size~$O(k\log n)$. Similar to other work on authenticated data structures~\cite{Merkle80,Merkle89,Goodrich01},
the proof reveals information about the underlying set, e.g., its size and the location
of the queried entry w.r.t.~other entries.

The model of a \emph{zero knowledge set} (more generally, \emph{zero knowledge elementary database}) was first introduced by Micali~\textit{et al.}~\cite{MicaliRK03}.
This is a secure data structure which allows a prover to commit to a \emph{finite set}
$S$ in such a way that, later on, it will be able to efficiently (and non-interactively)
prove statements of the form $x \in S$ or $x \notin S$ without leaking any information about $S$ beyond what has been queried for, not even the size of~$S$. 
The security properties  guarantee that the prover should not be able to cheat and prove contradictory statements about an element.
Later, Chase~\textit{et al.}~\cite{ChaseHLMR05} abstracted away Micali~\textit{et al.}'s solution and described the exact properties a commitment scheme should possess in order to allow a similar construction. 
This work introduced a new commitment scheme, called \emph{mercurial commitment}.
A generalization of mercurial commitments allowing for committing to an ordered sequence of messages ($q$-trapdoor mercurial commitment) was proposed in \cite{Catalano:2008} and 
later improved in \cite{Libert2010}. A~$q$-trapdoor mercurial commitment allows a committer to commit to an ordered sequence of message and 
later open messages with respect to specific positions.

Th above zero knowledge set constructions \cite{MicaliRK03,ChaseHLMR05,Catalano:2008,Libert2010} use an implicit
ordered $q$-way hash tree ($q \geq 2$) built
on the universe of all possible elements. The size of this tree is exponential in the security parameter. However, only a portion of the tree of size polynomial in the security parameter is explicitly stored in the data structure.
Let $N$ denote the universe size. Then the proof size for membership and non-membership for an individual element is $O(\log_q N)$. Kate~\textit{et al.}~\cite{KateZG10} suggested a weaker primitive called \emph{nearly-zero knowledge set} based on \emph{polynomial
commitment}~\cite{KateZG10}. In their construction the proof size for membership and non-membership for every individual element is $O(1)$, but the set size is not private. 

A related notion of \emph{vector commitments} was introduced by~\cite{CatalanoF13} where they show that
a (concise) $q$-trapdoor mercurial commitment can 
be obtained from a vector commitment and a trapdoor mercurial commitment.
A vector commitment scheme allows a committer
to commit to an ordered sequence of values $(x_1, \ldots, x_n)$ in such a way
that the committer can later open the commitment at specific positions (e.g., prove that $x_i$ is the $i$-th committed message).

Ostrovsky~\textit{et al.}~\cite{Ostrovsky04} generalized the idea of membership queries to support membership and orthogonal
range queries on a multidimensional dataset.~\cite{Ostrovsky04} describe constructions for consistent database
queries, which allow the prover to commit to a database, and then provide query answers that are provably consistent with the commitment.
They also consider the problem of adding privacy to such protocols. However their construction requires interaction
(which can be avoided in the random oracle model) and requires the prover to keep a counter for the questions asked so far.
The use of NP-reductions and probabilistically checkable proofs makes their generic construction expensive.
The authors of~\cite{Ostrovsky04} also provide a simpler protocol based
on \emph{explicit-hash Merkle Tree}.
However, this construction does not hide the size of the
database as the proof size is $O(\lceil \log n \rceil)$
where $n$ is the upper bound on the size of the database.

\paragraph{Signature Schemes}
 A collection of signature schemes, namely \emph{content extraction signature} \cite{Steinfeld01}, \emph{redactable signature} \cite{Johnson02} and 
\emph{digital document sanitizing scheme} \cite{Miyazaki06}
can be viewed in a three-party model where the owner digitally signs a data document
and the server discloses to the client only part of the signed document
along with a legitimately derived signature on it. The server derives the signature without the owner's involvement and
the client verifies the authenticity of the document it receives from the server by running the verification algorithm of the underlying scheme.
A related concept is that of \emph{transitive signature scheme}, where given the signatures of two edges $(a,b)$ and $(b,c)$
of a graph, it is possible to compute the signature for the edge (or path) $(a,c)$ without the signer's secret key \cite{MR2002,Yi:2007:DTS:2174147.2174160,CH2012}.
However, these signature schemes are not designed to preserve privacy of the signed object, which may include the content and/or the structure
in which the content is stored.

Ahn \textit{et al.} \cite{quoting} present a unified framework for computing on authenticated data via the notion of slightly homomorphic or $P$-homomorphic signatures, which was later improved by \cite{Wang12}. 
This broad class of $P$-homomorphic signatures includes \emph{quotable, arithmetic, redactable, homomorphic, sanitizable and transitive signatures}.
This framework allows a third party to derive a signature on the object $x^{\prime}$ from a signature on $x$ as long as $P(x,x^{\prime}) = 1$ for some predicate~$P$ that captures the \emph{authenticatable relationship} between $x$ and $x^{\prime}$. 
A derived signature reveals no extra information about the parent~$x$,
referred to as \emph{strong context hiding}.
This work does not consider predicates of a specific data structure.

The authors propose a general RSA-accumulator based scheme
that is expensive in terms of computation.
In particular,~the cost of signing depends on the predicate $P$ and the size of the message space and is $O(n^2)$ for a $n$-symbol message space.
This privacy definition was recently refined by \cite{AttrapadungLP12}.
This line of work cannot be directly used for privacy preserving data structures
where efficiency is an important requirement and quadratic
overhead may be prohibitive depending on the application.

\cite{ChaseKLM13} gives definition and construction of malleable signature scheme. 
A signature scheme is defined to be malleable if, given a signature $\sigma$ on a message
$x$, it is possible to efficiently derive a signature $\sigma'$ on a message $x'$ such that
$x' = T(x)$ for an \emph{allowable} transformation $T$.
Their definition of context hiding requires unlinkability and allows for adversarially-generated keys and signatures. This definition is stronger than that of \cite{quoting}
as it allows for adversarially-generated keys and signatures. 
Unlinkability implies the following: a quoted (or derived) signature should be indistinguishable from a fresh signature.

A motivating example proposed in \cite{quoting} deals with the impossibility of linking
a quote to its source document.
However, in the framework of privacy preserving authenticated data structures,
it is important for the client to verify membership, i.e., given
a quote from a document and a signature on the quote,
the client should be able to verify that the quote is indeed in the document. Context-hiding definition in \cite{ChaseKLM13} also requires unlinkability.

\paragraph{Leakage-Free Signature Schemes for Ordered Lists}

A \emph{leakage-free redactable} \emph{signature scheme (LRSS)} allows a third party to remove parts of a
signed document without invalidating its signature. This action, called \emph{redaction}, does not require
the signer's involvement. As a result, the verifier
only sees the remaining redacted document and is able to verify that
it is valid and authentic. Moreover, the redacted document and its signature do
not reveal anything about the content  or position of the removed parts.
This problem can be easily interpreted in the privacy-preserving authenticated data structure model, where the signer is the owner,
the third party is the server and the verifier is the client.

Kundu and Bertino \cite{KunduB08} were first to introduce
the idea of structural signatures for ordered trees (subsuming ordered lists) which support public redaction
of subtrees (by third-parties) while retaining the integrity of the remaining parts. This was later extended to DAGs and graphs \cite{KunduB13}.
The notion was later formalized as \emph{LRSS} for ordered trees
in \cite{Brzuska10} and subsequently several attacks on \cite{KunduB08} were also proposed in \cite{Brzuska10, Poehls12}.

The authors of~\cite{Chang09} presented a
leakage-free redactable signature scheme for strings (which can be viewed as an ordered list)
that hides the location of the redacted or deleted portions of the list
at the expense of quadratic verification cost.

The basic idea of the LRSS scheme presented in \cite{Brzuska10} is to
sign \emph{all possible ordered pairs} of elements of an ordered list.
So both the computation cost and the storage space are quadratic in the number of elements of the list.
Building on the work of~\cite{Brzuska10}, \cite{Samelin12} proposed an LRSS for lists
that has quadratic time and space complexity.
Poehls~\textit{et al.}~\cite{Poehls12} 
presented a LRSS scheme for a list that has linear time and space complexity
but assumes an associative non-abelian hash function, whose existence has not been formally proved. 
The authors of~\cite{Kundu12},  presented a construction that uses quadratic space at the server and is not 
leakage-free. We discuss the attack in Section~\ref{sec:scheme}.


\subsection{Contributions and Organization of the Paper}

The main contributions of this work are as follows:
\begin{itemize}[noitemsep,nolistsep]

\item After reviewing preliminary concepts and the cryptographic primitives we use in this paper,
in Section~\ref{sec:prelim}, we introduce the Zero-Knowledge List (ZKL) model, present a construction, prove its security and analyze its efficiency in Section~\ref{sec:ZKL}. 
\item In Section~\ref{sec:scheme}, we introduce a formal model for a privacy-preserving authenticated list that supports order queries on its elements.
\item In Section~\ref{sec:construction}, we present a construction of the above data structure based on bilinear maps and we analyze its performance.
\item Formal proofs for the security properties of our construction are given in Section~\ref{sec: sec-proofs}.

\end{itemize}

In Table~\ref{tab:comparison} we compare our constructions of a privacy-preserving authenticated list with previous work in terms of performance, and assumptions. 
We also indicate which constructions satisfy the zero-knowledge property.
We include a construction based on our new primitive, ZKL, and our direct construction of PPAL. We note that ZKL model is a two party model
but can be adapted to a three party model of PPAL (see Section~\ref{sec:construction} for details).
Our PPAL construction outperforms all previous work that is based on widely accepted assumptions \cite{Brzuska10,Samelin12}.

\begin{table}[t!]

  \scriptsize


  \begin{tabularx}{6.5in}{p{0.843in}>{\raggedright\arraybackslash} m{1.4cm} m{1.1cm} m{1cm} m{1cm} m{1cm} m{1.2cm} m{1cm}|  m{0.74cm} m{0.8cm}}
    &  &   &    &  & &  &  &  \multicolumn{2}{l}{This paper}               \\
    & \cite{Steinfeld01} & \cite{Johnson02}         & \cite{Chang09}       & \cite{Brzuska10}                & \cite{Samelin12} & \cite{Poehls12}                       & \cite{Kundu12} &  ZKL & \ppadsl              \\
    \hline
    Zero-knowledge     &                    &                          &                      & $\checkmark$                    & $~\checkmark$     & $~\checkmark$                          &      & $\checkmark$           & $\checkmark$                \\ 
    \hline
    Setup time        & $n \log n $        & $~n$                      & $~n$                  & $n^2$                           & $~n^2$            & $~n$                                   & $~n$         & $n \log N$       & $n$                         \\ 
    \hline
    Server Space             & $n$                & $~n$                      & $~n$                  & $n^2$                           & $~n^2$            & $~n$                                   & $~n^2$       & $n \log N$   & $n$                         \\ 
    \hline
    Query time        & $m$                & $~n \log n$           & $~n$                  & $mn$                            & $~m$              & $~n$                                 & $~n$  & $m\log N$& $\min(m\log n,n)$           \\
    \hline
    Verification time & $m\log n \log m$   & $~m\log n$                & $~n^2$             & $m^2$                           & $~m^2$            & $~m$                                   & $~m$    & $m \log N$        & $m$                         \\
    \hline
    
    Proof size       & $m$                & $~m\log n$                 & $~n$                  & $m^2$                           & $~m^2$            & $~m$                                   & $~n$        & $m \log N$    &  $m$                         \\
    \hline
    Assumption        & RSA                & ~RSA                      &  SRSA, Division & EUCMA  &ROH, nEAE          & AnAHF  & ROH, RSA  & ROH, FC, SRSA     & ~ROH,nBDHI \\   
    \hline
  \end{tabularx}
  
\small

\vspace{-10pt} 

\caption{\label{tab:comparison}%
  \small
  Comparison of our constructions of a privacy-preserving authenticated list with previous work.
  ZKL is a construction based on Zero-Knowledge lists from Section~\ref{sec:ZKLC}
  and \ppadsl~is a direct PPAL construction from Section~\ref{sec:construction} . All the time and space  complexities are asymptotic. Notation: 
  $n$ is the number of elements of the list, $m$ is the number of elements in the query, and 
  $N$ is the number of all possible $l$-bit strings from where list elements can be drawn from.
  Acronyms for the assumptions: 
  Strong RSA Assumption (SRSA);
  Existential Unforgeability under Chosen Message Attack (EUCMA) of the underlying signature scheme; 
   Random Oracle Hypothesis (ROH);
   $n$-Element Aggregate Extraction Assumption (nEAE);
  Associative non-abelian hash function (AnAHF);
   Factoring a composite (FC);
   $n$-Bilinear Diffie Hellman Inversion Assumption(nBDHI).
 }
\end{table}

\section{Preliminaries}
\label{sec:prelim}

\subsection{Data Type}

We consider a \emph{linearly ordered list}~$\llist$ as a data structure
that the owner wishes to store with the server. A list is an ordered set of elements 
$\llist = \{ x_1,x_2,\ldots,x_n \}$, where each $x_i \in \{ 0,1 \}^\ast$, $\forall x_1, x_2 \in \llist, x_1 \neq x_2$
and either $x_1 < x_2$ or $x_2 < x_1$.
Hence,~$<$~is a strict order on elements of $\llist$ that is irreflexive, asymmetric and transitive. 

We denote the set of
elements of the list \li~as \ele(\li).
A~sublist of $\llist$,
$\sublist$, is defined as: $\sublist = \lbrace x~|~x \in \elem(\llist)  \rbrace$.
Note that the order of elements in~$\sublist$ may not follow the order of~$\llist$.
We denote with \perm{\mathcal{L}}{\sublist} the permutation of the elements of $\sublist$ under the order of~\li.

$\llist(x_i)$ denotes the membership of element $x_i$ in $\llist$, i.e., $\llist(x_i) \neq \bot$ if $x_i \in \llist$ and $\llist(x_i) = \bot$ if $x_i \notin \llist$.
We interpret $\llist(x_i)$ as a boolean value, i.e., $\llist(x_i) \neq \bot$ is equivalent to $\llist(x_i) = \mathsf{true}$ and $\llist(x_i) = \bot$
is equivalent to $\llist(x_i) = \mathsf{false}$.
For all $x_i$ such that $\llist(x_i) \neq \bot$, $\mathsf{rank}(\llist, x_i)$ denotes the rank of element $x_i$ in the list, $\llist$.

\subsection{Cryptographic Primitives}

We now describe a signature scheme that is used in our construction
and cryptographic assumptions that underly the security of our method.
In particular, our zero knowledge list construction relies on
homomorphic integer commitments (Section~\ref{sssec:IntCom}),
zero knowledge protocol to prove a number is non-negative (Section~\ref{sssec:ZKP}) and zero
knowledge sets (Section~\ref{sssec:ZKS}),
while the construction for privacy preserving lists relies on bilinear aggregate signatures
and $n$-Bilinear Diffie Hellman Inversion assumption (Section~\ref{sssec:nbdhi}).

\subsubsection{Homomorphic Integer Commitment Scheme}
\label{sssec:IntCom}

We use a homomorphic integer commitment scheme $\HomIntCom$
that is statistically hiding and computationally binding~\cite{Boudot00, DamgardF02}.
The later implies the existence of a trapdoor and, hence, can be used to
``equivocate'' a commitment, that is open the original message of the commitment
to another message.
The above commitment scheme is defined in terms of three algorithms
$\HomIntCom = \{\cS,$ $\IC,$ $\IO\}$ and the corresponding trapdoor commitment (we call it a simulator) as:
$\Sim\HomIntCom = \{\Sim\cS,$ $\Sim\IC,$ $\Sim\IO\}$.
We describe these algorithms in Figure~\ref{fig:IntCom}.
The \emph{homomorphism} of $\HomIntCom$  is defined as 
$\IC(x+y) =  \IC(x) \times \IC(y)$.
For specific constructions of $\HomIntCom$ see Figure~\ref{fig:IC} in Appendix.

\begin{figure}[h]
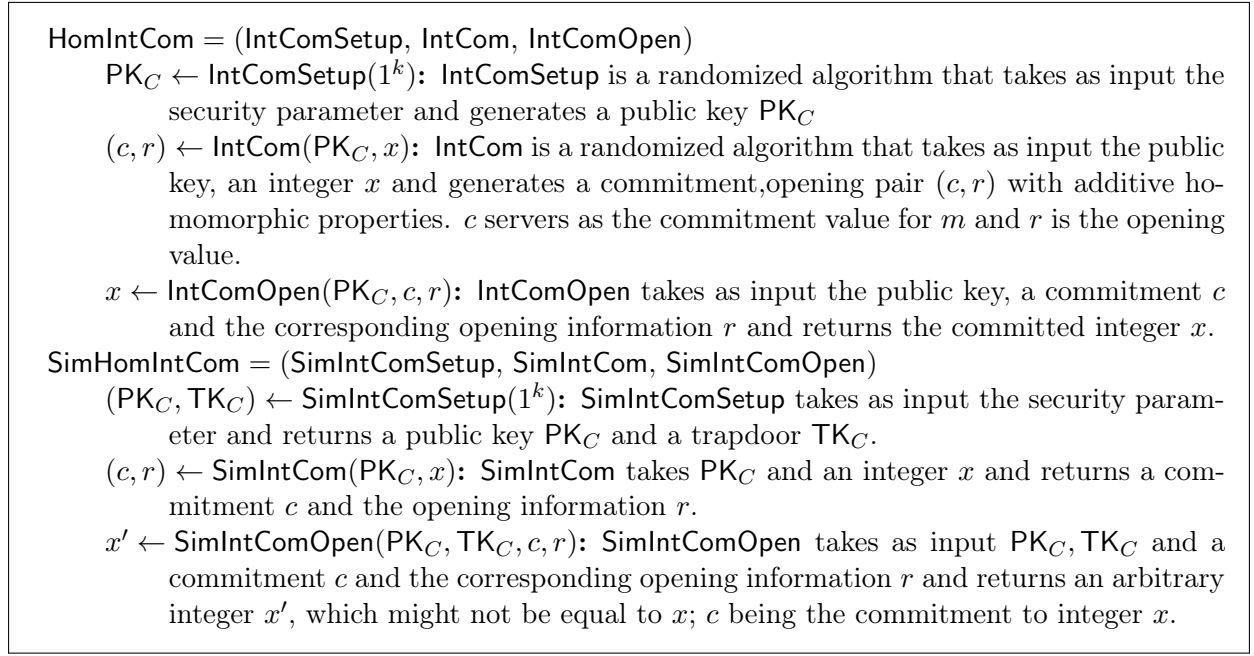

\caption{Homomorphic Integer Commitment Protocols.}
\label{fig:homintcom}
\begin{framed}
\begin{description}[noitemsep,nolistsep]
\item $\HomIntCom = (\cS,$ $\IC,$ $\IO)$ 
\begin{description}[noitemsep,nolistsep]
 \item [$\PK_C \leftarrow \cS(1^k)$:] 
$\cS$ is a randomized algorithm that takes as input the security parameter and generates a public key $\PK_C$
 \item[$(c,r) \leftarrow \IC(\PK_C,x)$:]
$\IC$ is a randomized algorithm that takes as input the public key, an integer~$x$ and generates a commitment,opening pair~$(c,r)$ with
additive homomorphic properties. $c$ servers as the commitment value for $m$ and $r$ is the opening value.
 \item [$x \leftarrow \IO(\PK_C,c,r)$:]
$\IO$ takes as input the public key, a commitment~$c$ and the corresponding opening information~$r$ and returns the committed integer~$x$.
\end{description}
\item $\Sim\HomIntCom = (\Sim\cS,$ $\Sim\IC,$ $\Sim\IO)$
\begin{description}[noitemsep,nolistsep]
 \item [$(\PK_C, \TK_C) \leftarrow \Sim\cS(1^k)$:] 
 $\Sim\cS$ takes as input the security parameter and returns a public key $\PK_C$ and a trapdoor $\TK_C$.

 \item[$(c,r) \leftarrow \Sim\IC(\PK_C,x)$:]
$\Sim\IC$ takes $\PK_C$ and an integer $x$ and returns a commitment $c$ and the opening information $r$.
 \item [$x' \leftarrow \Sim\IO(\PK_C,\TK_C,c,r)$:]
$\Sim\IO$ takes as input $\PK_C, \TK_C$ and a commitment $c$ and the corresponding opening information $r$ and returns an arbitrary integer $x'$, which
might not be equal to $x$; $c$ being the commitment to integer $x$.
\end{description}
\end{description}

\end{framed}
\label{fig:IntCom}
\end{figure}

\subsubsection{Proving an integer is positive in zero-knowledge}
\label{sssec:ZKP}
We use following protocol between a prover and a verifier:
the verifier holds prover's commitment~$c$ to an integer~$x$ and wishes to verify if this
integer is positive, $x > 0$,
without opening~$c$. We denote this protocol as $\Pro(x,r : c = C(x;r) \wedge x>0)$ (Figure~\ref{fig:ZKP}).
In our construction, we will use the commitment scheme~$\HomIntCom$ described in Figure~\ref{fig:IntCom} and
use $\IC$ to compute $c$.

\begin{figure}[h]
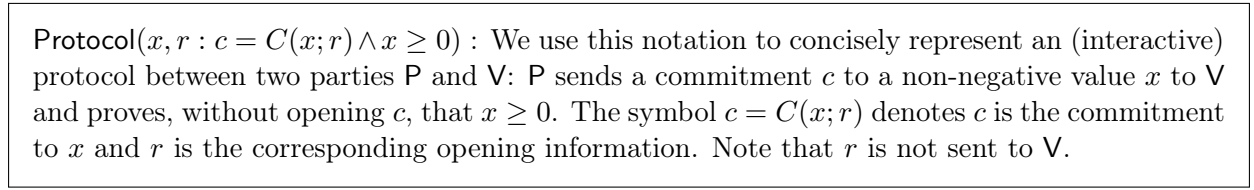

\caption{Protocol to prove non-negativity of an integer}
 \begin{framed}
$\Pro(x,r :c = C(x;r) \wedge x \geq 0)$ : We use this notation to concisely represent an (interactive) protocol between two parties $\mathsf{P}$ and $\mathsf{V}$:
$\mathsf{P}$ sends a commitment $c$ to a non-negative value $x$ to $\mathsf{V}$ and proves, without opening $c$, that $x \geq 0$.
The symbol $c = C(x;r)$ denotes $c$ is the commitment to $x$ and $r$ is the corresponding opening information. Note that $r$ is not sent to $\mathsf{V}$.
 \end{framed}
\label{fig:ZKP}
\end{figure}

As a concrete construction we extend
the protocol of~\cite{Lipmaa03} which allows one to prove that
$x \ge 0$ to supply a prove that $x-1 \ge 0$. This proves $x>0$.
The protocol of~\cite{Lipmaa03} is a $\Sigma$ protocol, which is \emph{honest verifier zero knowledge} and can be made 
\emph{non-interactive general zero knowledge}
in the Random Oracle model using Fiat-Shamir heuristic \cite{FiatS86}.
For details of the protocol refer to Figure~\ref{fig:Lipmaa}.

\subsubsection{Zero Knowledge Set scheme}
\label{sssec:ZKS}
Let $D$ be a set of of key value pairs. If $(x,v)$ is a key, value pair of $D$, i.e, $(x,v) \in D$,
then we write $D(x) = v$ to denote $v$ is the value corresponding to the key $x$. For the keys that are not present in~$D$,
$x \notin D$, we write $D(x) = \bot$.
A Zero Knowledge Set scheme (ZKS) consists of three probabilistic polynomial time algorithms - $\zks = (\zks\setupM,\zks\ZP = (\zks \Pone, \zks \Ptwo),\zks\ZV)$ and queries are of the form ``is key $x$ in $D$?''.
We describe the algorithms in Figure~\ref{fig:ZKSModel}.

\begin{figure}
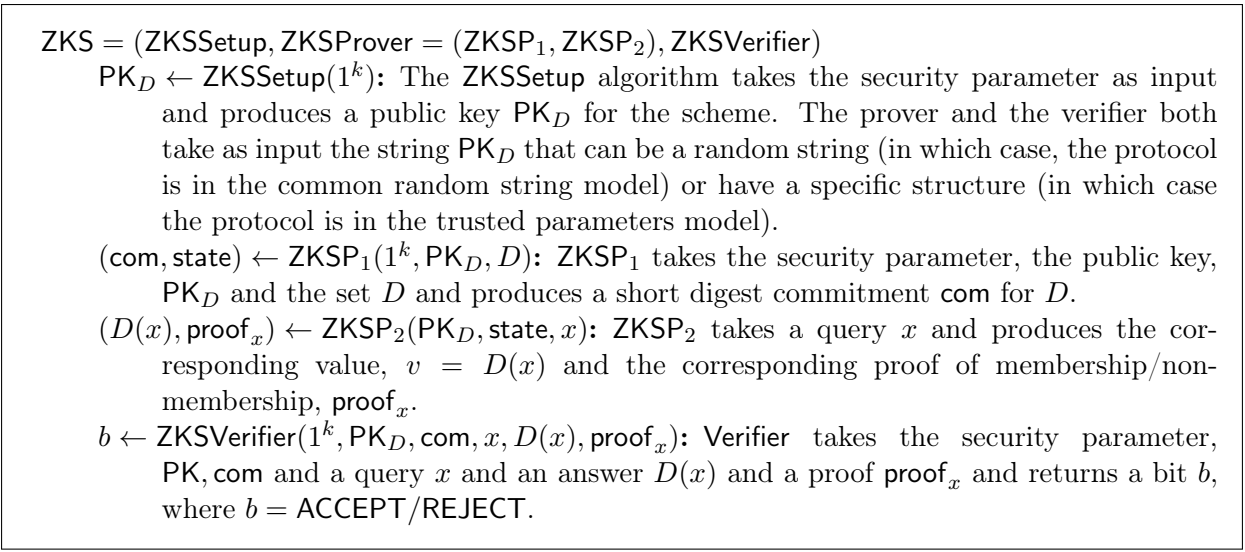

\caption{Zero Knowledge Set (ZKS) model}
\begin{framed}
\begin{description}[nolistsep,noitemsep]
 \item $\zks = (\zks\setupM,\zks\ZP = (\zks \Pone, \zks \Ptwo),\zks\ZV)$

\begin{description}[noitemsep,nolistsep]
 \item [$\PK_D \leftarrow \zks\setupM(1^k)$:] 
 
 The $\zks\setupM$ algorithm takes the security parameter as input and produces a public key $\PK_D$ for the scheme.
The prover and the verifier both take as input the string $\PK_D$ that can be a random
string (in which case, the protocol is in the common random string model) or
have a specific structure (in which case the protocol is in the trusted parameters model).

 \item[$(\com,\st) \leftarrow \zks \Pone(1^k, \PK_D, D)$:]
 
 $\zks \Pone$ takes the security parameter, the public key, $\PK_D$ and the set $D$ and produces a short digest commitment $\com$ for $D$.
 
 \item [$(D(x),\pr_x) \leftarrow \zks \Ptwo(\PK_D,\st, x)$:]
 
 $\zks \Ptwo$ takes a query $x$ and produces the corresponding value, $v = D(x)$ and the corresponding proof of membership/non-membership, $\pr_x$.

 \item [$b \leftarrow \zks \ZV(1^k, \PK_D, \com,x,D(x),\pr_x)$:]
 
 $\ZV$ takes the security parameter, $\PK,\com$ and a query $x$ and an answer $D(x)$ and a proof $\pr_x$ and returns a bit $b$, where $b= \mathsf{ACCEPT}/\mathsf{REJECT}$.
\end{description}
\end{description}
\end{framed}
\label{fig:ZKSModel}
\end{figure}

For our construction of zero knowledge lists we pick a ZKS construction of~\cite{ChaseHLMR05} that
is based on mercurial commitments and describe it in more details
in Figure~\ref{fig:ZKS}.

\subsubsection{Bilinear Aggregate Signature Scheme}
\label{sec:agg-sign}

We use bilinear aggregate signature scheme developed by Boneh~\textit{et al.}~\cite{boneh2003aggregate} for our privacy preserving authenticated data structure scheme. 
Given $n$ signatures on $n$ \emph{distinct} messages $M_1,M_2, \ldots, M_n$ from $n$ distinct users $u_1,u_2, \ldots, u_n$, it is possible to aggregate all these signatures into a single short signature such that 
 the single signature (and the $n$ messages) will convince the verifier that the $n$ users indeed signed the $n$ original messages (i.e., user $i$ signed message $M_i$ for $i = 1, \ldots, n$). 
Here we describe the scheme for the case of a single user signing $n$ \emph{distinct} messages $M_1,M_2, \ldots, M_n$. The decryption of the generic case of $n$ different users can be found at~\cite{boneh2003aggregate}.
The following notation is used in the scheme:

\begin{itemize}[noitemsep,nolistsep]
 \item $G,G_1$ are multiplicative cyclic groups of prime order $p$
 \item $g$ is a generator of $G$
 \item $e$ is computable bilinear nondegenerate map $e: G \times G \rightarrow G_1$
 \item $H: \lbrace 0,1 \rbrace^{\ast} \rightarrow G$ is a full domain hash function viewed as a random oracle that can be instantiated with a cryptographic hash function. 
\end{itemize}

Formally, a bilinear aggregate signature scheme is a 5 tuple of algorithm \emph{Key Generation, Signing, Verification, Aggregation}, 
and \emph{Aggregate Verification}. We discuss the construction in Figure~\ref{fig:BLGS}.

\begin{figure}
\caption{Bilinear Aggregate Signature Scheme}
\begin{framed}

\begin{description}[noitemsep,nolistsep]
 \item \textbf{Key Generation}: The secret key~$v$ is a random element of $\ZZ_p$ and the public key~$x$ is set to~$g^v$.
 
 \item \textbf{Signing}: The user signs the hash of each \emph{distinct message} $M_i \in \lbrace 0,1 \rbrace^{\ast}$ via $\sigma_i \leftarrow H(M_i)^v$.
 
 \item \textbf{Verification}: Given the user's public key $x$, a message $M_i$ and its signature~$\sigma_i$, accept if $e(\sigma_i,g) = e(H(M_i),x)$ holds.
 
 \item \textbf{Aggregation}: This is a public algorithm which does not need the user's secret key to aggregate the individual signatures. 
 Let $\sigma_i$ be the signature on a distinct message $M_i \in \lbrace 0,1 \rbrace^{\ast}$ by the user, according to the Signing algorithm ($ i = 1, \ldots, n$).
 The aggregate signature~$\sigma$ for a subset of $k$ signatures, where $k \leq n$, is produced via~$\sigma \leftarrow \prod_{i=1}^k \sigma_i$.
 
 \item \textbf{Aggregate Verification}: Given the aggregate signature $\sigma$, $k$ original messages $M_1,M_2, \ldots, M_k$ and the public key $x$:
        \begin{enumerate}[noitemsep,nolistsep]
         \item ensure that all messages $M_i$ are distinct, and reject otherwise.
         \item accept if $e(\sigma,g) = e(\prod_{i=1}^kH(M_i),x)$.
        \end{enumerate}

\end{description}

\end{framed}
\label{fig:BLGS}
\end{figure}

\paragraph{Security}
Informally, the security requirement of an aggregate signature scheme
guarantees that the aggregate signature~$\sigma$ is valid if and only if the aggregator
used all $\sigma_i$'s, for $1 \leq i \leq k$, to construct it. The formal model of security is called
the aggregate chosen-key security model. The security of aggregate signature schemes is
expressed via a game where an adversary is challenged to
forge an aggregate signature:

\begin{description}[noitemsep,nolistsep]
 \item \textbf{Setup}: The adversary $\mathcal{A}$ is provided with a public key $\mathsf{PK}$ of
 the aggregate signature scheme.
 
 \item \textbf{Query}: $\mathcal{A}$ adaptively requests signatures on messages of his choice.
 
 \item \textbf{Response}: Finally, $\mathcal{A}$ outputs $k$ distinct messages $M_1,M_2, \ldots, M_k$ and an aggregate signature~$\sigma$. 
\end{description}
$\mathcal{A}$ wins if the aggregate signature $\sigma$ is a valid aggregate signature on messages $M_1,M_2, \ldots, M_k$
under~$\mathsf{PK}$, and $\sigma$ is nontrivial, i.e., $\mathcal{A}$ did not request a signature on $M_1,M_2, \ldots, M_k$
under $\mathsf{PK}$. A formal definition and a corresponding security proof
of the scheme can be found in \cite{boneh2003aggregate}.


\subsection{Hardness assumption}
\label{sssec:nbdhi}

Let~$p$ be a large $k$-bit prime where
$k \in \mathbb{N}$ is a security parameter. Let $n \in \mathbb{N}$ be polynomial in $k$, $n = \p$.
Let $e:G \times G \rightarrow G_1$ be a bilinear map where $G$ and~$G_1$ are groups of prime order $p$
and $g$ be a random generator of $G$. We denote a probabilistic polynomial time~(PPT)
adversary~$\mathcal{A}$, or sometimes $\mathcal{B}$,  as an adversary
who is running in time $\mathsf{{poly}(k)}$.
We use $\mathcal{A}^{\mathsf{alg}(\mathsf{input}, \ldots)}$ to show that
an adversary $\mathcal{A}$ has an oracle access to
an instantiation of an algorithm~$\mathsf{alg}$ with
first argument set to~$\mathsf{input}$ and $\ldots$
denoting that $\mathcal{A}$ can give arbitrary input for the
rest of the arguments.

\begin{Definition}[$n$-Bilinear Diffie Hellman Inversion ($n$-BDHI)~\cite{Boneh04efficientselective-id}]\label{def:nbdhi}
Let $s$ be a random element of $\mathbb{Z}_p^\ast$ and $n$ be a positive integer.
Then, for every PPT adversary~$\mathcal{A}$ there exists a negligible function $\nu(.)$ such that:
$Pr[s \xleftarrow{\$}  \mathbb{Z}_p^\ast, y \leftarrow G_1:\mathcal{A}(\langle g, g^s,g^{s^2}, \ldots, g^{s^n} \rangle): y  = e(g,g)^{\frac{1}{s}}] \leq \nu(k).$
\end{Definition}


\section{Zero Knowledge List (ZKL)}
\label{sec:ZKL}

We generalize the idea of consistent set membership queries~\cite{MicaliRK03,ChaseHLMR05} to support membership and order queries in \emph{zero knowledge} on a list with \emph{no repeated elements}.
More specifically, given a totally ordered list of unique
elements~$\llist = \{ y_1,y_2,\ldots,y_n \}$, we want to support in zero knowledge queries of the following form: 
\begin{itemize}[nolistsep,noitemsep]
 \item Is $y_i \in \llist$ or $y_i \notin \llist$? 
 \item For two elements $y_i, y_j \in \llist$, what is their relative order, i.e., $y_i <y_j$ or $y_j <y_i$ in $\llist$?
\end{itemize}
We adopt the same adversarial model as
in~\cite{MicaliRK03,Ostrovsky04,ChaseHLMR05}. Thus, we require that proofs reveal nothing beyond the query answer,
not even the size of the list.
There are two parties: the \emph{prover} and the \emph{verifier}.
The \emph{prover} initially commits to a list of values and makes the commitment (a short digest) public.
Informally, the security properties can be stated as follows.
Completeness mandates that honestly generated proofs always satisfy the verification test.
Soundness states that the prover should not be able to come up with a query,
 and corresponding inconsistent (with the initial commitment) answers and convincing proofs.
Finally, zero-knowledge means that each proof reveals the answer and nothing else. In other words, there
must exist a simulator, that given only an oracle access to $\llist$,
can
simulate proofs for membership and order queries that are indistinguishable from real proofs.
Next, we formally describe the model and the security properties.

\subsection{Model}
\label{ssec:ZKLModel}

A Zero Knowledge List scheme (ZKL) consists of three probabilistic polynomial time algorithms - $(\ZS,\ZP = (\Pone, \Ptwo),\ZV)$ and the queries are of the form $(\sublist,\f)$ where $\sublist = \{ z_1, \ldots, z_m \}$, $z_i \in \{ 0,1 \}^\ast$,  is a collection of elements,
$\f=0$ denotes a membership/non-membership query and
$\f=1$ denotes an order query. In the following sections, we will use $\st$ to 
represent a variable that saves the current state of the algorithm (when it finishes execution). 

\begin{description}

\item $\PK \leftarrow \ZS (1^k)$

The $\ZS$ algorithm takes the security parameter as input and produces a public key $\PK$ for the scheme.
The prover and the verifier both take as input the string $\PK$ that can be a random
string (in which case, the protocol is in the common random string model) or
have a specific structure (in which case the protocol is in the trusted parameters model).

\item $ (\com, \st) \leftarrow \Pone(1^k, \PK, \llist)$

$\Pone$ takes the security parameter, the public key $\PK$ and the list $\llist$, and produces a short digest commitment $\com$ for the list.

\item $(\mem,\pr_M, \Or, \pr_O) \leftarrow \Ptwo(\PK,\st,\sublist,\f)$  where $\sublist = \{z_1, \ldots, z_m \}$ and
$\f$ denotes the type of query.
$\Ptwo$ produces the membership information of the queried elements, $\mem = \{ \llist(z_1), \ldots, \llist(z_m) \}$ and the proof of membership (and non-membership), $\pr_M$. Then depending on $\f$:
\begin{description}
\item [$\f = 0$:]
$\Ptwo$ sets $\Or$ and $\pr_O$ to $\bot$ and returns
$(\mem,\pr_M, \bot, \bot)$.
\item [$\f = 1$:]
Let $\tilde{\sublist} = \{ z_i \mid i \in [1,m] \wedge \llist(z_i) \neq \bot \}$.
$\Ptwo$ produces the correct list order among the elements of $\tilde{\sublist}$, $\Or = \permM{\llist}{\tilde{\delta}}$ and the proof of the order, $\pr_O$.
\end{description}

\item $b \leftarrow \ZV(1^k,\PK,\com,\sublist,\f,\mem, \pr_M, \Or, \pr_O)$

$\ZV$ takes the security parameter, the public key $\PK$, the commitment $\com$ and a query $(\sublist,\f)$ and $\mem, \pr_M, \Or, \pr_O$ and returns a bit $b$, where $b= \mathsf{ACCEPT}/\mathsf{REJECT}$.

\end{description}

\paragraph{Example}
Let us illustrate the above functionality with a small example. Let $\llist = \{A, B, C \}$ and $(\sublist, \f) = (\{B, D, A\}, 1)$ be the query.
Then given this query $\Ptwo$ returns $\mem = \{ \llist(B),$ $\llist(D),$ $\llist(A)\}=  
\{\mathsf{true}, \mathsf{false}, \mathsf{true} \}$, the corresponding proofs of membership and non-membership in $\pr_M$,
$\Or = \{ A, B \}$  and the corresponding proof of order between $A$ and $B$ in $\pr_O$. 


\subsection{Security Properties}
\label{ssec:ZKLSec}

\begin{Definition}[Completeness]
\label{def:ZKL-Completeness}
For every list $\llist$, every sublist~$\sublist$ and every $\f$,

\begin{alignat*}{2}
 \Pr[\PK \leftarrow \ZS(1^k); &(\com,\st) \leftarrow \Pone(1^k, \PK, \llist);\\    
 & (\mem, \pr_M, \Or, \pr_O) \leftarrow \Ptwo(\PK,\st,\sublist,\f):\\                     
    & \ZV(1^k,\PK,\com,\sublist,\f,\mem,\pr_M, \Or, \pr_O) &=&  \mathsf{ACCEPT}] = 1         
 \end{alignat*}

\end{Definition}

\begin{Definition}[Soundness]
\label{def:ZKL-Soundness}
For every PPT malicious prover algorithm, $\ZP '$,
for every sublist~$\sublist$ and for every~$\f$ there exists a negligible function $\nu(.)$ such that:

\begin{alignat*}{2}
 \Pr[&\PK \leftarrow \ZS(1^k); \\
 &(\com, \mem^1, \pr_M^1, \Or^1, \pr_O^1,  \mem^2, \pr_M^2, \Or^2, \pr_O^2) \leftarrow \ZP'(1^k,\PK):\\
& \ZV(1^k,\PK,\com,\sublist,\f, \mem^1, \pr_M^1, \Or^1, \pr_O^1) = \mathsf{ACCEPT} {\huge\wedge} \\
 &\ZV(1^k,\PK,\com,\sublist,\f, \mem^2, \pr_M^2, \Or^2, \pr_O^2) = \mathsf{ACCEPT} \wedge \\
 &((\mem^1 \neq \mem^2) \vee (\Or^1 \neq \Or^2)) ] \leq \nu(k)   
 \end{alignat*}
\end{Definition}

\begin{Definition}[Zero-Knowledge]
\label{def:ZKL-ZK}

There exists a PPT simulator $\Sim = (\Sim_1, \Sim_2, \Sim_3)$ such that for every PPT malicious verifier $\advN = (\advN_1, \advN_2)$, there exists a negligible function $\nu(.)$ such that: 
\begin{alignat*}{2}
 |\Pr[\PK \leftarrow \ZS(1^k);(\llist, \st_A)  \leftarrow \advN_1(1^k,\PK); &(\com, \st_P) \leftarrow \Pone (1^k, \PK, \llist):\\
 &\advN_2^{\Ptwo(\PK,\st_P,\cdot)}(\com,\st_A) = 1] - \\
 \Pr[(\PK,\st_S) \leftarrow \Sim_1(1^k);(\llist, \st_A)  \leftarrow \advN_1(1^k, \PK); &(\com, \st_S) \leftarrow \Sim_2(1^k,\st_S):\\
& \advN_2^{\Sim_3^{\llist}(1^k,\st_S)}(\com,\st_A) = 1] | \leq \nu(k)
 \end{alignat*}

Here $\Sim_3$ has oracle access to $\llist$, that is,
given a query $(\sublist,\f)$, $\Sim_3$ can query the list $\llist$ to learn only the membership/non-membership
of elements in $\sublist$
and, if $\f=1$, learn the list order of the elements of $\sublist$ in $\llist$.
\end{Definition}



\subsection{Zero Knowledge List (ZKL) Construction}
\label{sec:ZKLC}

\paragraph{Intuition}
The construction uses zero knowledge set scheme, homomorphic integer commitment scheme, zero-knowledge protocol to prove non-negativity of an integer and 
a collision resistant hash function $\mathbb{H}:\{0,1\}^\ast \rightarrow \{0,1\}^l$, if the elements of the list $\llist$ are larger that $l$ bits.
In particular, given an input list $\llist$ the prover creates a set $D$ where for every element $y_j\in \llist$
it adds a (key,value) pair $(\mathbb{H}(y_j),C(j))$ where
$\mathbb{H}(y_j)$ is a hash of $y_j$ and $C(j)$ is a homomorphic integer commitment of $\mathsf{rank}(\llist,y_j)$ (assuming $\mathsf{rank}(\llist,y_j)=j$
without loss of generality).
The prover then sets up a zero knowledge set on $D$ using $\zks \Pone$ from zero knowledge set construction in Figure~\ref{fig:ZKSModel}.
The output of $\zks \Pone$ is a commitment to $D$, $\com$, that the prover sends to the verifier.

Membership and non-membership queries of the form $(\sublist, 0)$ are replied in the same fashion as in zero knowledge set, by invoking $\zks \Ptwo$
on the hash of every element of sublist $\sublist$.
Recall that as a response to a membership query for a key, $\zks \Ptwo$ returns the value against that key.
In our case, the queried key is $\mathbb{H}(y_j)$ and the value returned by $\zks \Ptwo$,~$D(\mathbb{H}(y_j))$ is the commitment $C(j)$
where $j$ is the rank of element $y_j$ in the list $\llist$, if $y_j \in \llist$. If $y_j \notin \llist$, the value returned is $\bot$.
Hence, the verifier receives the commitments to ranks for queried member elements. These commitments are never opened but are used
as part of a proof for order queries.

For a given order query $(\sublist,1)$, for every adjacent pair of elements in the returned order, $\Or$, the prover
gives a proof of order. Recall that $\Or$ contains the member elements of $\sublist$, arranged according to their order in the list, $\llist$.
To prove the order between two elements $y_i, y_j$, the prover does the following.
Let $\mathsf{rank}(\llist,y_i) = i,\mathsf{rank}(\llist,y_j) = j$,
and $C(i)$, $C(j)$ the corresponding commitments and wlog $i < j$.
As noted above, $C(i), C(j)$ are already returned by the prover as a part of membership proof.
Additionally, the prover augments the membership proof with a commitment to 1, $C(1)$, and its
opening information $\rho$.

Then the verifier computes $C(j-i-1) := C(j)/(C(i)C(1))$ using
the homomorphic property of the integer commitment scheme.
The prover and the verifier then engage in $\Pro(x,r :c = C(x;r) \wedge x \geq 0)$ to convince the verifier that
$C(j-i-1)$ is a commitment to value $x=j-i-1 \ge 0$.
Note that we use the non-interactive zero-knowledge version of the protocol as discussed in Section~\ref{sssec:ZKP}.

It is important to understand why we require $\ZV$ to verify that $j-i-1 \ge 0$ and not $j-i\ge 0$.
By the soundness of the protocol $\Pro(x,r :c = C(x;r) \wedge x \geq 0)$, the probability that  a cheating prover $\ZP'$ will be able to convince $\ZV$ about the non-negativity of a negative integer is negligibly small.
However, since $0$ is non-negative, a cheating prover can do the following:
instead of the rank of an element store the same arbitrary non-negative integer for every element in the list.
Then, $C(j-i)$ and $C(i-j)$ are commitments to $0$ and $\ZP'$ can always succeed in proving an arbitrary order.
To avoid this attack, we require the prove to hold for~$C(j-i-1)$.
An honest prover can always prove
the non-negativity of $C(j-i-1)$ as $|j-i|\geq 1$ for any rank $i,j$ of the list.

Also, we note that the commitments to ranks can be replaced by commitments to a strictly monotonic sequence as long as there is a
1:1 correspondence with the rank sequence. 
In this case, the distance between two elements will also be positive and, hence, the above protocol still holds.

\paragraph{Construction}

Let $\HomIntCom = (\cS,\IC,\IO)$ be the homomorphic integer commitment scheme defined in Section~\ref{sssec:IntCom} and $\zks = (\zks\setupM,$ $\zks\ZP = (\zks \Pone,$ $\zks \Ptwo),\zks\ZV)$ be a ZKS scheme
defined in Section~\ref{sssec:ZKS}.
We denote the output of the prover during the non-interactive statistical zero knowledge protocol $\Pro(x,r :c = C(x;r) \wedge x \geq 0)$ as $\pr_{x\geq0}$.
The construction also uses a hash function, $\mathbb{H}: \{0,1 \}^\ast \rightarrow \{0,1\}^{l}$.
In Figure~\ref{fig:ZKLScheme} we describe in detail our ZKL construction on an input list
$\llist = \{ y_1, \ldots, y_n \}$.

\begin{figure}[ht!]
\caption{Zero Knowledge List (ZKL) Construction}
\begin{framed}

\begin{description}[nolistsep,noitemsep]
 \item [$\PK \leftarrow \bm {\ZS} (1^k)$:]
 
The $\ZS$ algorithm takes the security parameter as input and runs $ \PK_C \leftarrow \cS(1^k), \PK_D \leftarrow \zks\setupM(1^k)$ and outputs $\PK = (\PK_C,\PK_D)$.

\item [$(\com, \st) \leftarrow \bm {\Pone}(1^k, \PK, \llist)$:]

Wlog, let $\rank(\llist,y_j) = j$ and $C(j)$ denote an integer commitment to $j$ under public key $\PK_C$, i.e., $(C(j),r_j) = \IC(\PK_C, j)$.
Then, $\Pone$ proceeds as follows:

\begin{itemize}[noitemsep,nolistsep]
 \item For every $y_j \in \llist$, compute $\mathbb{H}(y_j)$ and $C(j)$.
 \item Set $D: = \{ (\mathbb{H}(y_j), C(j)) \mid \forall y_j \in \llist \}$.
 \item Run $(\com,\st) \leftarrow \zks\Pone(1^k,\PK_D,D)$ and output $(\com,\st)$.
\end{itemize}

\item [$(\mem,\pr_M, \Or, \pr_O) \leftarrow \bm {\Ptwo}(\PK,\st,\sublist,\f)$ where $\sublist = \{ z_1, \ldots, z_m \}$:]
Let $S := \{ \mathbb{H}(z_1), \ldots, \mathbb{H}(z_m) \}$. For all $x \in S$ do the following:
 \begin{itemize}[nolistsep,noitemsep]
  \item Run $(D(x),\pr_x) \leftarrow \zks\Ptwo(\PK_D,\st,x)$.
  \item Set $\Delta_x : = (D(x), \pr_x)$.
 \end{itemize}
Set $\mem := \{ \llist(z_j) \mid \forall z_j \in \sublist \}$ and $\pr_M := \{ \Delta_x \mid x \in S \}$.\\ 
If $\f = 0$ return $(\mem,\pr_M,\bot,\bot)$.\\
If $\f=1$ do the following:\\
Let $\tilde{\sublist} = \{ z_j \mid \forall j \in [1,m] \wedge \llist(z_j) \neq \bot \}$ and $\permM{\llist}{\tilde{\delta}} =  \{ w_1, \ldots, w_{m'} \}$ where $m' \leq m$.
\begin{itemize}[nolistsep,noitemsep]
 \item For all $1 \leq j < m'$, compute $\Delta_{w_{j} <w_{j+1}} = \pr_{\rank(\llist,w_{j+1}) - \rank(\llist,w_j) -1 \geq  0}$.
 \item  Compute $(C(1),\rho) = \IC(\PK_C,1)$.
\end{itemize}
Set $\Or := \permM{\llist}{\tilde{\delta}}$ and $\pr_O = (\{\Delta_{w_{j}<w_{j+1}} \mid (w_j, w_{j+1}) \in \tilde{\delta}\},C(1),\rho)$ and return $(\mem,\pr_M,\Or,\pr_O)$.

\item[ $b \leftarrow \bm{\ZV}(1^k,\PK,\com,\sublist,\f,\mem, \pr_M, \Or, \pr_O)$ where $\sublist = \{ z_1, \ldots, z_m \}$:]

The $\ZV$ algorithm does the following:
\begin{itemize}[noitemsep,nolistsep]
 \item Compute $S = \{ \mathbb{H}(z_1), \ldots, \mathbb{H}(z_m) \}$. 
 \item Parse $\pr_M$ as $\pr_M := \{ \Delta_x = (D(x), \pr_x) \mid x \in S \}$.
 \item For all $x \in S$, run $b \leftarrow \zks\ZV(1^k, \PK_D,x,D(x),\pr_x)$.
\end{itemize}

If $\f=0$ and $b = \mathsf{ACCEPT}$ for all  $x \in S$, output $\mathsf{ACCEPT}$.\\
If $\f=1$, perform the following additional verification steps:

\begin{itemize}[nolistsep,noitemsep]
 \item Let $\Or = \{ w_1, \ldots, w_{m'} \}$.
 \item Parse $\pr_O$ as  $(\{\Delta_{w_{j}<w_{j+1}} \mid (w_j, w_{j+1}) \in \Or\},C(1),\rho)$.
 \item Verify that $\IO(\PK_C, C(1),\rho)$ is 1.
 \item Compute $D(\mathbb{H}(w_{j+1})) / (D(\mathbb{H}(w_j))\times C(1) ) = C(\rank(\llist,w_{j+1}) - \rank(\llist,w_j) - 1)$
 \item Verify that $\rank(\llist,j+1) - \rank(\llist,j) >  0$ using $\pr_{\rank(\llist,j+1) - \rank(\llist,j) - 1 \geq  0}$ using
 $\Pro(x,r :c = C(x;r) \wedge x \geq 0)$ where $x = \rank(\llist,j+1) - \rank(\llist,j)-1$.
\end{itemize}

  If all the verifications pass, only then return $\mathsf{ACCEPT}$.

\end{description}

\end{framed}
\label{fig:ZKLScheme}
\end{figure}



\subsection{Security Proofs}


\paragraph{Proof of Completeness}

Completeness of the ZKL construction in Section~\ref{sec:ZKLC} directly follows from the Completeness of Zero Knowledge Set and Completeness of the protocol  $\Pro(x,r :c = C(x;r) \wedge x \geq 0)$. \qed

\paragraph{Proof of Soundness:}

To simplify the notation, first let us denote using $\Ev_1$ and $\Ev_2$ the following two events:
\begin{align*}
 &\Ev_1 = [\PK \leftarrow \ZS(1^k); \\
 &(\com, \mem^1, \pr_M^1, \Or^1, \pr_O^1,  \mem^2, \pr_M^2, \Or^2, \pr_O^2) \leftarrow \ZP'(1^k,\PK):\\
 &\ZV(1^k,\PK,\com,\sublist,\f, \mem^1, \pr_M^1, \Or^1, \pr_O^1) = \mathsf{ACCEPT} \wedge \\
 &\ZV(1^k,\PK,\com,\sublist,\f, \mem^2, \pr_M^2, \Or^2, \pr_O^2) = \mathsf{ACCEPT} \wedge \\
 &(\mem^1 \neq \mem^2)]
\end{align*}
\begin{align*}
& \Ev_2 = [\PK \leftarrow \ZS(1^k); \\
& (\com, \mem^1, \pr_M^1, \Or^1, \pr_O^1,  \mem^2, \pr_M^2, \Or^2, \pr_O^2) \leftarrow \ZP'(1^k,\PK):\\
& \ZV(1^k,\PK,\com,\sublist,\f, \mem^1, \pr_M^1, \Or^1, \pr_O^1) = \mathsf{ACCEPT} \wedge \\
& \ZV(1^k,\PK,\com,\sublist,\f, \mem^2, \pr_M^2, \Or^2, \pr_O^2) = \mathsf{ACCEPT} \wedge \\
& (\Or^1 \neq \Or^2)]
\end{align*}
Then, Definition~\ref{def:ZKL-Soundness} can be rewritten as
\begin{align*}
& \Pr[\PK \leftarrow \ZS(1^k); \\
& (\com, \mem^1, \pr_M^1, \Or^1, \pr_O^1,  \mem^2, \pr_M^2, \Or^2, \pr_O^2) \leftarrow \ZP'(1^k,\PK):\\
& \ZV(1^k,\PK,\com,\sublist,\f, \mem^1, \pr_M^1, \Or^1, \pr_O^1) = \mathsf{ACCEPT} \wedge \\
& \ZV(1^k,\PK,\com,\sublist,\f, \mem^2, \pr_M^2, \Or^2, \pr_O^2) = \mathsf{ACCEPT} \wedge \\
& ((\mem^1 \neq \mem^2) \vee (\Or^1 \neq \Or^2)) ] = \Pr[\Ev_1 \vee \Ev_2] \leq \Pr[\Ev_1] + \Pr[\Ev_2]
\end{align*}

Now, by the Soundness property of the ZKS in Section~\ref{sssec:ZKS}, $\Pr[\Ev_1]$ is negligible in $k$. Let $\Pr[\Ev_1] = \nu_1(k)$.

Let us consider the event $\Ev_2$. If the malicious prover is successful in outputting two contradictory orders for a 
collection of elements, then there must exist at least one inversion pair, i.e., a pair of elements
$(x_i,x_j) \in \sublist$ such that $x_i < x_j$ in $\Or^1$ and $x_j <x_i$ in $\Or^2$.
Let $C(i)$ and $C(j)$ be the commitments used  as values to prove
the membership of $x_i$ and $x_j$, correspondingly.
Then by the binding property of the integer commitment scheme of Section~\ref{sssec:IntCom}, $\ZP'$ cannot equivocate
$C(i-j)$ or $C(j-i)$ (which is computed by $\ZV$ in the protocol).
(Note that by the soundness property of ZKS, the probability that $\ZP'$
can return two commitments $C(i)$ and $C(i')$, $C(i) \neq C(i')$, where $C(i)$
and $C(i')$ are returned 
to prove membership of $x_i$ in $\pr_M^1$ and $\pr_M^2$, respectively, is negligible
w.r.t.~the same
commitment, $\com$.)
Then according to the protocol, it must be the case that $\ZP'$ could convince $\ZV$ that both $C(i-j)$ and $C(j-i)$ are commitments to positive integers where $i,j$ are two integers.
However, due to the soundness of the protocol $\Pro(x,r :c = C(x;r) \wedge x \geq 0)$, the probability is negligible in $k$. Let $\Pr[\Ev_2] = \nu_2(k)$.

Therefore we have, $\Pr[\Ev_1 \vee \Ev_2] \leq \nu_1(k) + \nu_2(k) \leq \nu(k)$, for some negligible function $\nu(.)$
Hence the soundness error of the ZKL construction must be negligible in $k$. \qed


\paragraph{Proof of Zero-Knowledge:}

Let $\Sim\HomIntCom = (\Sim\cS,$ $\Sim\IC,$ $\Sim\IO)$ be the simulator of $\HomIntCom$ defined in Figure~\ref{fig:homintcom}.
Let $\Sim\zks = (\Sim\zks\setupM,\Sim\zks\ZP = (\Sim\zks \Pone,$ $\Sim\zks \Ptwo),$ $\Sim\zks\ZV)$ be the simulator for the ZKS in Figure~\ref{fig:ZKSModel}.

Now let us define $\Sim = (\Sim_1,\Sim_2,\Sim_3)$, a simulator for ZKL (Definition~\ref{def:ZKL-ZK}),
that has access to the system parameter $\mathbb{H}$. 
\begin{itemize}[noitemsep]
 \item $\Sim_1(1^k)$ runs $(\PK_D,\TK_D) \leftarrow \sZS(1^k)$ and $(\PK_C,\TK_C)  \leftarrow \Sim\cS(1^k)$.
 $\Sim_1(1^k)$ outputs $\{ \PK = (\PK_D,\PK_C), \TK = (\TK_D,\TK_C) \}$.
 \item $\Sim_2$ runs $\Sim\zks \Pone$ to generate commitment $\com$.
 \item In response to membership queries $(\f = 0)$, $\Sim_3$ does the following: 
\begin{itemize}[nolistsep,noitemsep]
 \item $\Sim_3$ maintains a table of queried elements as tuples $\langle x_i, v_i,r_i \rangle$ where $x_i$ is the queried element
 and $v_i$ is the value that $\Sim_3$ has sent when $x_i$ was queried. We explain how $r_i$ is computed next. 
 \item For a queried element $y$, $\Sim_3$ checks the table.
  If $y$ is not in the table and, hence, has not been queried before,
 $\Sim_3$ makes an oracle access to $\llist$ on $y$.
 If $y \in \llist$, $\Sim_3$ computes a fresh commitment to 0, $(C(0),r) := \Sim\IC(0)$,
  and stores $\langle y, C(0),r\rangle$. If $y \notin \llist$, then $\Sim_3$ stores $\langle y, \bot, \bot \rangle$.  
 \item $\Sim_3$ responds to membership queries by invoking $\Sim\zks \Ptwo$ on $\mathbb{H}(y)$ and returning the same output.
\end{itemize}
 
 \item For order queries ($\f = 1$), $\Sim_3$ additionally does the following.
 Let $\sublist$ be the queried sublist. $\Sim_3$ makes an oracle access to $\llist$ to get the correct list order of the elements of $\sublist$
 that are present in $\llist$. Let
 $\Or = \{ y_1, \ldots, y_m \}$ be the returned order.
 \item $\Sim_3$ computes $(C(1),\rho) = \Sim\IC(\PK_C,1)$.
 \item Let $\{\langle y_1, v_1, r_1\rangle, \ldots, \langle y_m, v_m,r_m\rangle\}$ be the entries of $\Sim_3$'s table that correspond to elements in $\sublist$.
 Then for every pair $(y_j, y_{j+1})$, $\Sim_3$ equivocates $(v_{j+1} / (v_j \times C(1)))$ using $\TK_C$ to a commitment to any arbitrary positive integer $u$.
 In other words, $\Sim_3$ equivocates the commitment $C(\rank(\llist,y_{j+1})-\rank(\llist,y_j)-1)$ to a commitment to an arbitrary positive integer $u$.
Finally, $\Sim_3$ computes $\pr_{u \geq 0}$ to prove the order between $(y_j, y_{j+1})$.

\end{itemize}

$\Sim_3$ achieves the following.
For every newly queried element that is in the list, $\Sim_3$ generates and stores a fresh commitment to $0$, and sends it to the verifier.
Hence, $\Sim_3$ sets $\rank = 0$ to all queried elements.
By the hiding property of the integer commitment scheme, the commitments are identically distributed to the commitments computed by the real prover, $\Pone$.
Now, with the help of $\TK_C$,~$\Sim_3$ can equivocate a commitment to any value it wants.
Hence, whenever he needs to prove order $y_i < y_j$, 
$\Sim_3$ equivocates the commitment to $\rank(\llist,y_{j+1})-\rank(\llist,y_j)-1$ to any arbitrary positive integer $u$ and invokes
the protocol $\Pro(u,r :c = C(u;r) \wedge u \geq 0)$  to compute $\pr_{u>0}$. 

Since the protocol $\Pro(u,r :c = C(u;r) \wedge u \geq 0)$ is Zero Knowledge (Statistical),
$\Sim = (\Sim_1,\Sim_2,\Sim_3)$ simulates our ZKL scheme. \qed

We note that the constructions with which we instantiate ZKL have the simulators assumed above.
In particular, for $\Sim\zks$ we use the simulator of the ZKS construction of~\cite{ChaseHLMR05}.
For $\Sim\HomIntCom$ we use the construction of~\cite{DamgardF02} and for completeness define
a simulator in Figure~\ref{fig:SimIC}.

\subsection{Efficiency}
\label{sssec:ZKLEfficiency}

The efficiency of our ZKL construction depends on the efficiency of the underlying constructions that we use.
We consider the the ZKS construction used in~\cite{ChaseHLMR05}  based on Mercurial Commitments,
the homomorphic integer commitment of~\cite{DamgardF02} and a protocol for non-negative proof of a commitment
from~\cite{Lipmaa03}. Each of these constructions is described in more detail in Appendix.
Mercurial commitment was later generalized by \cite{Catalano:2008,Libert2010} but the basic ZKS construction remains the same.

Recall that $k$ is the security parameter of the scheme, $l$ is the size of the output of the hash function $\mathbb{H}$,
$n$ is the number of elements in the list $\llist$ and $m$ is the number of elements in query~$\sublist$. Similarly to~\cite{ChaseHLMR05} 
we assume that $l=k$.
For every element in~$\llist$, $\Pone$ hashes the element and computes a commitment to its rank,
taking time $O(1)$. It then computes $n$ height-$k$ paths to compute
the commitment $\com$ to a list, $\llist$, takes time $O(kn)$, where $|\llist| = n$. For further details please see Appendix~\ref{sec:ZKSCons}.

Membership (non-membership) proof of a single element consists of $O(k)$ mercurial decommitments.
Using \cite{Libert2010}, we can have each mercurial decommitment constant size, i.e, $O(1)$.
The order proof between two elements requires membership proofs for both
elements and $\pr_{u-1 \geq 0}$ where $u$ is the absolute difference between the rank of the corresponding elements.
$\pr_{u-1 \geq 0}$ is computed using
$\Pro(x,r :c = C(x;r) \wedge x \geq 0)$ which takes $O(1)$ time.
Hence, computing a membership proof for a single element or an order
proof for two elements  takes time $O(k)$.
More generally, the prover's time for a query on sublist $\sublist$ is
$O(mk)$, where $m = |\sublist|$.

The verifier needs to verify $O(k)$ mercurial decommitments for every element in the query $\sublist$
and verify order between every adjacent pair of elements in $\sublist$ using $\Pro(u,r :c = C(u;r) \wedge u \geq 0)$.
Therefore, the asymptotic run time of the verification is $O(mk)$. 

We summarize the properties and efficiency of our ZKL construction in Theorem~\ref{thm:ZKL}.

\begin{Theorem}
\label{thm:ZKL}

The zero-knowledge list (ZKL) construction of
Figure~\ref{fig:ZKLScheme} satisfies the security
properties of completeness (Definition~\ref{def:ZKL-Completeness}),
soundness (Definition~\ref{def:ZKL-Soundness}) and zero-knowledge
(Definition~\ref{def:ZKL-ZK}). 
The construction has the following performance, where  $n$ is the list size, $m$ is the
query size, each element of the list is a $k$-bit\footnote{If not, we can use a hash function to reduce every element to a $k$-bit string,
as shown in the construction} string and $N$ is the number
of all possible $k$-bit strings.
\begin{itemize} \itemsep 0pt
\item The prover executes the setup phase in $O(n \log N)$ time and space.
\item In the query phase, the prover computes the proof of the answer
  to a query in $O(m \log N)$ time.
\item The verifier verifies the proof in  $O(m \log N)$ time and space.
\end{itemize}

\end{Theorem}



\section{Privacy Preserving Authenticated List (PPAL)}
\label{sec:scheme}

In the previous section we presented a model and a construction for a new primitive called
zero knowledge lists. As we noticed earlier, ZKL model gives the desired functionality
to verify order queries on lists. However, the corresponding construction does not provide the efficiency one may
desire in cloud computing setting where the verifier (client) has limited memory resources.
In this section we address this setting and define a model
for privacy preserving authenticated lists, \ppadsl, that is executed between
three parties. This model, arguably, fits cloud scenario better and as we will see our construction
is also more efficient. In particular, the size of a single proof in \ppadsl~is
$O(1)$ vs. $O(k)$ in ZKL.

\subsection{Model}

\ppadsl~is a tuple of three probabilistic polynomial time algorithms $(\setupM,\Q,\ve)$
executed between the owner of the data list $\llist$,
the server who stores $\llist$ and answers
queries from the client and the client who issues queries
and verifies corresponding answers.

\begin{description}

 \item $(\D, \K) \leftarrow \setupM (1^k, \llist)$
 
 This algorithm takes the security parameter and the source list~$\llist$ as input and produces two digests~$\D$ and~$\K$ for the list.
 This algorithm is run by the owner. $\D$ is sent to the client and $\K$ is sent to the server.

  \item $ (\Or, \pr) \leftarrow \Q(\K, \llist, \sublist$)
  
  This algorithm takes the key generated by the owner, $\K$, the source list, $\llist$ and a queried sublist, $\sublist$, as input, where a sublist of a list $\llist$ is defined as: \ele($\sublist$) $\subseteq$ \ele(\li).
  The algorithm produces the list order of the elements of $\llist$, $\Or = \permM{\mathcal{L}}{\delta}$, and a proof, $\pr$, of the answer. This algorithm is run by the server.

   \item $b \leftarrow \ve(\D, \sublist , \Or, \pr)$
   
   This algorithm takes $\D$, a queried sublist $\sublist$, $\Or$ and $\pr$ and returns a bit $b$, where $b= \mathsf{ACCEPT}$ iff \ele($\sublist$) $\subseteq$ \ele(\li) and  $\Or = \permM{\mathcal{L}}{\delta}$. Otherwise,
   $b= \mathsf{REJECT}$. This algorithm is run by the client.

\end{description}

\subsection{Security Properties}

A \ppadsl~has three important security properties. The first property is \emph{Completeness}. This property ensures that for any list~$\llist$ and for any sublist $\sublist$ of $\llist$,
if the $\D,\K,\Or,\pr$ are generated honestly, i.e., the owner and the server honestly execute the protocol, then the client will be always convinced about the correct list order of $\sublist$.

\begin{Definition}[Completeness]
\label{def:completeness}
For all lists $\llist$ and all sublists~$\sublist$
\begin{eqnarray*}
\Pr[(\D, \K) \leftarrow \setupM (1^k, \llist);(\Or, \pr) \leftarrow \Q(\K,\llist, \sublist): \\
\ve(\D, \sublist , \Or, \pr) = \mathsf{ACCEPT}] =1
\end{eqnarray*}

\end{Definition}
The second security property is \emph{Soundness}. This property ensures that once an honest owner generates a pair $(\D, \K)$
for a given list $\llist$,
even a malicious server 
will not be able to convince the client of incorrect order of elements belonging to the list $\llist$. This property ensures integrity of the scheme.

\begin{Definition}[Soundness]
 \label{def:soundness}
For all PPT malicious Query algorithms $\Q '$, for all lists $\llist$ and all query sublists~$\sublist$, there exists a negligible function $\nu(.)$ such that:
\begin{align*}
 \Pr[(\D, \K) \leftarrow \setupM (1^k, \llist);(\Or_1, \pr_1, \Or_2, \pr_2) &\leftarrow \Q'(\K,\llist): \\
 \ve(\D, \sublist , \Or_1, \pr_1) = \mathsf{ACCEPT}& \wedge \\
 \ve(\D, \sublist , \Or_2, \pr_2) = \mathsf{ACCEPT} &\wedge \\
 (\Or_1 \neq \Or_2&) ] \leq \nu(k)
\end{align*} 
 
\end{Definition}

The last property is \emph{Zero-Knowledge}. This property captures that even a malicious client cannot learn anything about the list (and its size)
beyond what the client has queried for.
Informally,
this property involves showing that there exists a simulator such that even for adversarially chosen list~$\llist$, no adversarial client (verifier) can tell if it is talking to an honest owner and server pair
who are committed to $\llist$ or to the simulator who only has oracle access to the list $\llist$.

\begin{Definition}[Zero-Knowledge]
\label{def:ZK}

There exists a PPT simulator $\Sim = (\Sim_1, \Sim_2)$ such that for all PPT malicious verifiers $\advN = (\advN_1, \advN_2)$, there exists a negligible function $\nu(.)$ such that:

\begin{align*}
 | \Pr[(\llist, \st_A)  \leftarrow \advN_1(1^k); &(\D, \K) \leftarrow \setupM (1^k, \llist):\\
 &\advN_2^{\Q(\K,\llist,.)}(\D,\st_A) = 1] - \\
 \Pr[(\llist, \st_A)  \leftarrow \advN_1(1^k); &(\D, \st_S) \leftarrow \Sim_1(1^k):\\
 &\advN_2^{\Sim_2^{\llist}(1^k,\st_S)}(\D,\st_A) = 1] | \leq \nu(k)
\end{align*}
 
\end{Definition}

Here $\Sim_2$ has oracle access to $\llist$, that is given a sublist $\sublist$ of $\llist$, $\Sim_2$ can query the list $\llist$ to learn only the correct list order of the sublist $\sublist$ and cannot look at $\llist$.

\paragraph{Attack on \cite{Kundu12}'s scheme}
We observer that the scheme presented in~\cite{Kundu12}
does not satisfy the zero knowledge property of \ppadsl~ for the following reason.
The scheme of~\cite{Kundu12} generates a $n'$ bit secure name,
where $n' \geq n$,
for each element of the list of size~$n$. 
A high level idea of the scheme is as follows.
The secure name of an element has dedicated bits, where
each bit corresponds to the pairwise order
between this element and every other element in the list.
To prove the order between any two elements, the verifier needs to know secure names for
both of them.
Then, given any two secure names, the verifier can easily compute the required bit.
Two order queries $A < B$ and $A<C$,
as per the scheme of \cite{Kundu12},
reveal to the client the secure names of all three elements $A$, $B$ and $C$.
Hence, given the secure names of $B$ and $C$, the client can easily compute the bit which preserves the order information between $B$ and $C$ and infer the order between
them. Therefore, it is impossible to write a simulator~$\Sim$ for an adversarially generated list such that the view of the adversary is indistinguishable as in Definition~\ref{def:ZK}.

\section{\ppadsl~Construction}
\label{sec:construction}

We present an implementation of a privacy preserving authenticated list in Figure~\ref{fig:protocol}.
We provide the intuition of our method followed
by a more detailed description.

\paragraph{Intuition}
Every element of the list is associated with a member
witness where a member witness is a randomized bilinear accumulator.
This allows us to encode the rank of the element (i.e., accumulate it) inside of the member witness
and then ``blind'' this rank information with randomness.
Every pair of (element, member witness) is signed by the owner
and the signatures are aggregated using bilinear aggregate signature scheme presented in Figure~\ref{fig:BLGS}, to compute the list digest signature.
Signatures and digest are sent to the server, who can use them to prove authenticity when answering client queries.
The advantage of using an aggregate signature is for the the server to be
able to compute a valid digest signature for any sublist of the source list
by exploiting the homomorphic nature
of aggregate signatures, that is
without owner's involvement.
Moreover, the client can verify the individual signatures in a single shot using aggregate signature verification.

The owner also sends linear (in the list size) number of random elements used in the encoding of member witnesses.
These random elements
allow the server to compute the order witnesses on queried elements, without the owner's involvement.
The order witness encodes the distance between two elements, i.e., the difference between element ranks,
without revealing anything about it.
Together with randomized accumulators as member witnesses, the client
can later use bilinear map to verify the order of the elements.

\paragraph{Construction}
Our construction for $\ppads$ is presented in Figure~\ref{fig:protocol}.
It is based on bilinear accumulators and bilinear aggregate signature introduced in \cite{boneh2003aggregate}
and described here in Section~\ref{sec:agg-sign}.
We denote \emph{member witness}
for $x_i \in \llist$ as \authunitm{x_i}.
For two elements $x_i, x_j \in $ \li, such that $x_i < x_j$ in \li $\;$, $\authunitoM{x_i}{x_j}$ is an \emph{order witness} for the order between $x_i$ and $x_j$.

\afterpage{%
\thispagestyle{empty}
\begin{figure}
\caption{Privacy-Preserving Authenticated List (PPAL) Construction}
\begin{framed}

\textbf{Notation}: $k \in \mathbb{N}$ is the security parameter of the scheme; $G,G_1$ multiplicative cyclic groups of prime order $p$ where $p$ is large $k$-bit prime; $g$:
a random generator of $G$; $e$: computable bilinear nondegenerate map $e: G \times G \rightarrow G_1$;
 $\mathcal{H}: \lbrace 0,1 \rbrace^{\ast} \rightarrow G$: full domain hash function
 (instantiated with a cryptographic hash function);
all arithmetic operations are performed using$\mod p$.
$\llist$ is the input list of size~$n = \p$, where $x_i$ are distinct
and~$\mathsf{rank}(\llist,x_i) = i$.
System parameters are $(p,G,G_1,e,g,\mathcal{H})$.

\begin{description}[noitemsep,nolistsep]

\item

$(\D,\K) \leftarrow {\bm \setupM}(1^k,\llist)$, where
\begin{description}[noitemsep,nolistsep]
 \item $\llist$ is the input list of length~$n$;
 \item $\D = (g^\x, \sigma_\llist)$;
 \item $\K = (g^\x, \sigma_\llist\ , \langle g, g^s, g^{s^2}, \ldots, g^{s^{n}} \rangle, \Sigma_\llist, \R )$ and
 \begin{description}[noitemsep,nolistsep]
 \item $ \langle s\xleftarrow{\$} \ZZ_p^{\ast}, \x\xleftarrow{\$} \ZZ_p^{\ast}\rangle$ is the secret key of the owner;
 \item $\Sigma_\llist = \langle \lbrace  {{t_{x_i \in \llist}, \sigma_i}} {\rbrace}_{1 \leq i \leq n}, \mathcal{H}({\omega}) \rangle$ is member authentication information and
 $\omega$ is the list $\mathsf{nonce}$;
 \item $\R = \langle r_1,r_2, \ldots, r_n \rangle, r_i \neq r_j$ for $i \neq j$, is order authentication informations;
 \item $\sigma_\llist$ is the digest signature of the list $\llist$.
 \end{description}
\end{description}
These elements are computed as follows:
\begin{description}[noitemsep,nolistsep]
 \item For every element~$x_i$ in $\llist = \{x_1, \ldots, x_n\}$: Pick $r_i \xleftarrow{\$} \ZZ_p^{\ast}$. Compute member witness for index $i$ as $t_{x_i \in \llist} \leftarrow ({g^{s^i}})^{r_i} $
  and signature for element $x_i$ as ${\sigma_i} \leftarrow {\mathcal{H}(t_{x_i \in \llist}||x_i)}^\x $.
 \item Pick the $\mathsf{nonce}$, ${\omega} \xleftarrow{\$} {\lbrace 0,1 \rbrace}^{\ast}$, which should be unique for each list.
 \item Set $ \salt \leftarrow {({\mathcal{H}}({{\omega}})})^\x$. $\salt$ is treated as a list identifier which protects against mix-and-match attack and also protects from the leakage that the queried result is the complete list.
 \item The list digest signature is computed as: $ {\sigma}_\llist \leftarrow  \salt \times {\prod}_{1 \leq i \leq n} {\sigma_i}$.
\end{description}
\item $(\Or,\pr) \leftarrow {\bm \Q} (\K,\llist,\sublist)$, where
\begin{description} [noitemsep,nolistsep]
 \item $\sublist = \lbrace z_1, \ldots, z_m \rbrace$ s.t.~$z_i \in \llist, \; \forall i \in [1,m]$, is the queried sublist; 
 \item $\Or = \pi_{\llist}(\sublist) = \{y_1, y_2, \ldots, y_m\}$;
 \item $\pr = (\Sigma_\Or,\Omega_\Or)$:
 \begin{description} [noitemsep,nolistsep]
 \item $\Sigma_\Or = (\sigma_\Or,T, \lambda_{\llist'})$ where~$\llist^{\prime} = \llist \setminus {\sublist}$;
 \item $T = \lbrace  t_{y_1 \in \llist}, \ldots,  t_{y_m \in \llist} \rbrace$;
 \item $\Omega_\Or =  \{$\authunito{y_1}{y_{2}}, \authunito{y_2}{y_{3}}, \ldots,\authunito{y_{m-1}}{y_{m}}$\}$.
\end{description}
\end{description}

These elements are computed as follows:
\begin{description}[noitemsep,nolistsep]
 \item The digest signature for the sublist:  ${\sigma}_{\Or} \leftarrow {\prod}_{y_j \in \Or} {\sigma_{\mathsf{rank}(\llist, y_j)}}$.
 \item  The member verification unit: $\lambda_{\llist'} \leftarrow \mathcal{H}({\omega}) \times {\prod}_{{x \in \llist^\prime}} {\mathcal{H}(t_{x_{{\mathsf{rank}(\llist, x)}}\in \llist}||x)}$.
 \item For every $j \in [1,m-1]$:  Let $i' = \mathsf{rank}(\llist, y_j)$ and $i'' = \mathsf{rank}(\llist, y_{j+1})$, and $r' = \R[i']^{-1}$ and $r'' = \R[i'']$.
      Compute \authunito{y_j}{y_{j+1}} $ \leftarrow {(g^{s^d})}^{r'r''}$ where $d =|i'-i''|$. 
\end{description}

\item $b \leftarrow {\bm \ve} (\D,\sublist,\Or, \pr)$ where $\D,\sublist,\Or,\pr$ are defined as above.\\ 
The algorithm checks the following:
\begin{itemize}[noitemsep,nolistsep]
 \item Compute ${{\xi}} \leftarrow {\prod}_{y_j \in \sublist}  {\mathcal{H}}(t_{y_j \in \llist}||y_j)$ and check if $e({\sigma}_{\Or},g) \stackrel{?}{=} e({\xi},g^\x)$
 \item Check if $e({\sigma}_\llist,g) \stackrel{?}{=} e({\sigma}_{\Or},g) \times e(\lambda_{\llist'},g^\x)$
 \item For every $j \in [1, m-1]$, $e(t_{y_j \in \llist}, $\authunito{y_j}{y_{j+1}}$) \stackrel{?}{=} e($\authunitm{y_{j+1}}$, g)$
\end{itemize}
Return $\mathsf{ACCEPT}$ iff all the equalities of the three steps verify, $\mathsf{REJECT}$ otherwise.

\end{description}
\end{framed}
\label{fig:protocol}
\end{figure}
\clearpage
}

The construction works as follows.
In the $\setupM$ phase, the owner generates secret key $(\x,s)$ and public key $g^\x$, where $\x$ is used for signatures. 
The owner picks a distinct random element~$r_i$ from the group $\ZZ_p^\ast$ for each element $x_i$ in the list~$\llist$, $i \in [1,n]$.
The element~$r_i$ is used to compute the member
witness~$t_{x_i \in \llist}$.
Later in the protocol, together with~$r_j$,
it is also used by the server
to compute the order witness $t_{x_i < x_j}$ for $x_i$ and $x_j \in \llist$ where~$x_i < x_j$ in~$\llist$.
The owner also computes individual signatures,~$\sigma_i$'s,
for each element and aggregates them into a digest signature~$\sigma_\llist$ for the list. 
It returns the signatures and
member witnesses for every element of $\llist$ in $\Sig$
and the set of random numbers picked for each index
to be used in order witnesses in~$\R$.
The owner sends $\D = (g^\x, \sigma_\llist)$ to the client
and $\K = (g^\x, \sigma_\llist\ , \langle g, g^s, g^{s^2}, \ldots, g^{s^{n}} \rangle, \Sigma_\llist, \R )$ and $\llist$ to the server. 

Given a query~$\sublist$, the server returns a response list~$\Or$ that contains
elements of~$\sublist$ in the order they appear in~$\llist$.
The server uses information in~$\Sig$ to
build $\Sigma_\Or$ from member witnesses
of elements in~$\sublist$, and compute the digest signature~$\sigma_\sublist$
for $\sublist$ and its membership verification unit~$\lambda_{\llist'}$
where $\llist' = \llist \setminus \sublist$. The server uses information in~$\R$ to compute $\Rgen_\Or$.
The client first checks that all the returned elements are indeed signed by the owner using $\Sigma_\Or$ and then verifies the order of the returned elements using $\Rgen_\Or$.

\paragraph{Preprocessing at the Server}
For a query~$\sublist$ on the list~$\llist$
of length~$m$ and~$n$, respectively, the $\Q$~algorithm
in Figure~\ref{fig:protocol}
takes $O(m)$ time to compute~$\sigma_\delta$ and
$O(n-m)$ to compute~$\lambda_{\llist'}$.
The server can precompute and store some products to reduce the overall running
time of this algorithm to~$O(m\log n)$ when $m \ll n$.
The precomputation proceeds as follows.

Let $\psi_i = \mathcal{H}(t_{x_i \in \llist}||x_i)$ for every element in $\llist = \{x_1, \ldots, x_n\}$.
Then the precomputation proceeds by computing a balanced binary tree 
over $n$ leaves, where $i$th leave corresponds to~$x_i$ and stores~$\psi_i$.
Each internal node of the tree stores the product of its children.
Therefore the root stores the complete product  
${\prod}_{i=1}^n \psi_i$. (See Figure~\ref{fig:RT} for an illustration of the tree.)
Computing each internal node takes time $O(1)$ since at each internal node product of at most two children is computed.
Since the tree has $O(n)$ nodes,
the precomputation takes time $O(n)$ and requires~$O(n)$ storage.

\begin{figure}[H]
    \begin{center}
        \includegraphics[scale=0.4]{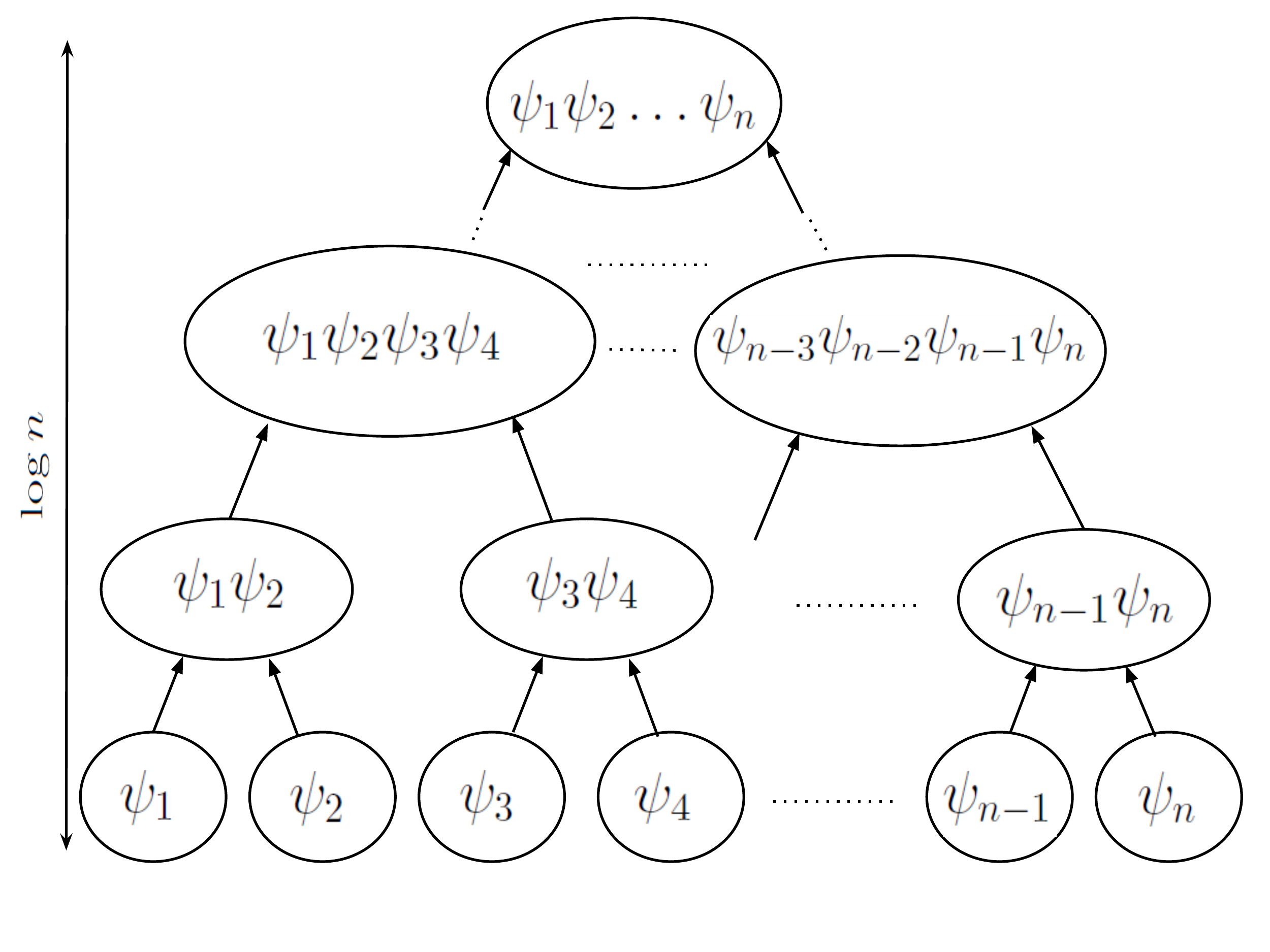}
    \caption{Range tree showing the precomputed products where $\psi_i = \mathcal{H}(t_{x_i \in \llist}||x_i)$.
    Precomputed products allow to speed up the computation time of $\Q$ algorithm in Figure~\ref{fig:protocol} when $m \ll n$.}
    \label{fig:RT}
    \label{default}
    \end{center}
\end{figure}

Now, computing $ \lambda_{\llist'}$ will require computing the product of $m+1$ partial products, i.e.,
the intervals between elements in the query. Since each partial product can be computed using at most $O(\log n)$ of the precomputed products (as the height of the tree is $O(\log n)$), 
the total time required to compute the product of $m+1$ partial products is $O((m+1) \log n) = O(m \log n)$.
Hence, the precomputation is useful whenever $m \ll n$.
Otherwise, when $m = O(n)$, the server can run the $\Q$~as mentioned in the scheme in Figure~\ref{fig:protocol}
in time $O(n)$.

\paragraph{Efficiency}

We measure the time and space complexity
of our scheme in terms of~$n$, the length of the list~$\llist$,
and~$m$, the length of the queried sublist~$\sublist$.
We use $|\llist|$ notation to denote the length of a list $\llist$.
Recall that \ele($\sublist$)  $\subseteq$ \ele($\llist$). 
We discuss and summarize the time and space complexity for each party as follows:

\begin{description}
 \item [Owner]The \setup~algorithm computes member and order witnesses for each element,
 along with signatures for each element.
 Hence, the algorithm runs in time $O(n)$ and requires~$O(n)$ space.

 \item [Server] Computing $\lambda_{\llist'}$ that takes time $O(n-m)$, as it touches $|\llist \setminus \query|$ elements
 and computing ${\sigma}_{\sublist}$ takes time $O(m)$.
 Hence, the overall runtime of computing $\lambda_{\llist'}$ and ${\sigma}_{\sublist}$~is $ O(n)$.
 The server can precompute and store some products of the signatures, as mentioned above,
 to reduce the overall running time to $O(\min \lbrace m\log n,n \rbrace)$.
 In addition the server calculates $m - 1$
 order witnesses each taking constant time, hence,
 $O(m)$ in total. So
 the overall run time for the server is $O(\min \lbrace m\log n,n \rbrace)$.
 The server needs to store the list itself, $\K$ and the precomputed products.
 Since each of these objects is of size $O(n)$, the space requirement
 at the server is~$O(n)$.
 
 \item [Client] $\ve$~computes a hash for each element in the
 query~$\sublist$,
 and then checks the first two equalities using bilinear map. 
 This requires~$O(m)$ computation. 
 In the last step $\ve$~checks $O(m)$ bilinear map equalities which takes time $O(m)$.
 Hence the overall verification time of the client is $O(m)$.
 During the query phase, the client requires
 $O(m)$ space to store its query and its response with the proof for
 verification. The client also needs to store $\D$ which requires $O(1)$ space.
\end{description}

\paragraph{Efficiency of PPAL via ZKL}

We noted in the introduction that we can adapt zero knowledge lists
to implement a PPAL scheme.
Recall that we can do this by making the owner run $\Pone$, the server run $\Ptwo$ and the client run $\ZV$ of
ZKL (see Section~\ref{ssec:ZKLModel} for the description of ZKL algorithms).
Here we estimate the efficiency of a PPAL construction based on the construction
of ZKL presented in~Figure~\ref{fig:ZKLScheme} and compare it with the \ppadsl~construction presented in this section.

From the discussion of efficiency of the ZKL construction in Section~\ref{sssec:ZKLEfficiency}, the time and space complexity of each party
in \ppadsl~adaptation of ZKL readily follows below.

\begin{description}[nolistsep,noitemsep]
 \item [Owner] The owner runs in time $O(kn)$ and $O(kn)$ space, where $k$ is the security parameter.

 \item [Server] To answer a query of size $m$, the server runs in time $O(km)$. The space requirement
 at the server is~$O(kn)$ since he has to store the $O(kn)$ commitments produced by the owner.
 \item [Client] The verification time of the client is $O(km)$.
 During the query phase, the client requires
 $O(km)$ space to store its query and its response with the proof for
 verification.
\end{description}

Hence, the \ppadsl~construction presented in Figure~\ref{fig:protocol} is a factor of $O(k)$
more efficient in space and computation requirements as compared to
an adaptation of the ZKL construction from~Figure~\ref{fig:ZKLScheme} in \ppadsl~model.

 \paragraph{Batch ordering query}
 The client can learn the total order among $m$ different elements of the list using a basic ordering query on two elements. This requires $O(m^2)$ individual order queries, where each verification
 takes one multiplication in group $G$ and six bilinear map computations.
  Since our construction supports a query of multiple elements,
  the client can optimize the process and ask a singe batch ordering query for $m$ elements.
  In this case, the verification will require only $m$ multiplications in the group~$G$ and
  $2m+2$ bilinear map computations.
%


\section{Security of the \ppadsl~Construction}
\label{sec: sec-proofs}

In this section we prove that the construction presented in Section~\ref{sec:construction}
is \ppadsl~construction according to definitions of completeness, soundness and zero knowledge
in Section~\ref{sec:scheme}.

\subsection{Proof of Completeness}

If all the parties are honest, all the equations in $\ve$ evaluate to true. This is easy to see just by expanding the equations as follows:

\begin{description}[nolistsep,noitemsep]
 \item [Equation $e({\sigma}_{\Or},g) \stackrel{?}{=} e({\xi},g^\x)$ :]
 
 Let $\Or = \{ y_1, \ldots, y_m \} = \pi_{\llist}(\sublist)$
 \begin{align*}
  e({\sigma}_{\Or},g) = &e({\prod}_{y_j \in \Or} {\sigma_{\mathsf{rank}(\llist, y_j)}},g) = e({\prod}_{y_j \in \Or}{\mathcal{H}(t_{y_j \in \llist}||y_i)}^\x,g) = \\ 
  &e({\prod}_{y_j \in \Or}{\mathcal{H}(t_{y_j \in \llist}||y_i)},g^\x) = e({\prod}_{y_j \in \sublist}{\mathcal{H}(t_{y_j \in \llist}||y_i)},g^\x) = e(\xi,g^\x).
 \end{align*}
 \item [Equation $e({\sigma}_\llist,g) \stackrel{?}{=} e({\sigma}_{\Or},g) \times e(\lambda_{\llist'},g^\x)$:] Let $\Or = \{ y_1, \ldots, y_m \} = \pi_{\llist}(\sublist)$ and $\llist' = \llist \setminus \sublist$.
 We start with the right hand side.
 \begin{eqnarray*}
  e({\sigma}_{\Or},g) \times e(\lambda_{\llist'},g^\x) = e({\prod}_{y_j \in \Or}{\mathcal{H}(t_{y_j \in \llist}||y_i)}^\x,g) \times 
  e(\mathcal{H}({\omega}) \times {\prod}_{{x \in \llist^\prime}} {\mathcal{H}(t_{x_{{\mathsf{rank}(\llist, x)}}\in \llist}||x)} ,g^\x)\\
  =  e({\prod}_{y_j \in \Or}{\mathcal{H}(t_{y_j \in \llist}||y_i)},g^\x) \times e(\mathcal{H}({\omega}) \times {\prod}_{{x \in \llist^\prime}} {\mathcal{H}(t_{x_{{\mathsf{rank}(\llist, x)}}\in \llist}||x)} ,g^\x)\\
  = e(\mathcal{H}({\omega}) \times {\prod}_{x \in \llist} {\mathcal{H}(t_{x_{{\mathsf{rank}(\llist, x)}}\in \llist}||x)}, g^\x) =  
  e(\mathcal{H}({\omega})^\x \times {\prod}_{x \in \llist} {\mathcal{H}(t_{x_{{\mathsf{rank}(\llist, x)}}\in \llist}||x)}^\x, g) = e({\sigma}_\llist,g).
\end{eqnarray*}
 \item [Equation $e(t_{y_j \in \llist}, $\authunito{y_j}{y_{j+1}}$) \stackrel{?}{=} e($\authunitm{y_{j+1}}$, g)$:]
 Let $i' = \mathsf{rank}(\llist, y_j)$ and $i'' = \mathsf{rank}(\llist, y_{j+1})$ and $r' = \R[i']^{-1}$ and $r'' = \R[i'']$. 
 \begin{eqnarray*}
  e(t_{y_j \in \llist},t_{y_j<y_{j+1}}) &=& e(g^{s^{i'} (r')^{-1}},g^{s^{i'' -i'} r'' r'})  = e(g,g)^{s^{i''-i'+i'} r'' r'(r')^{-1}}\\
  &=& e(g,g)^{s^{i''} r''} = e(g^{s^{i''} r''},g) = e(t_{y_{j+1} \in \llist},g).
 \end{eqnarray*}

\end{description}

\subsection{Proof of Soundness}

Soundness follows by reduction to the $n$-Bilinear Diffie Hellman assumption~(see Definition~\ref{def:nbdhi} for details).
To the contrary of the Soundness Definition~\ref{def:soundness},
assume that given a list~$\llist$, the malicious server, $\Q'$ produces two different orders $\Or_1 \neq \Or_2$ for some sublist $\sublist$
such that corresponding order proofs are accepted by the client, i.e., by algorithm~$\ve$ in Figure~\ref{fig:protocol}.
Let $\sublist = \lbrace x_1, x_2, \ldots, x_m \rbrace$.
Since $\Or_1 \neq \Or_2$, then there exists at least one inversion pair $(x_i,x_j)$ such that $x_i < x_j$ in $\Or_1$ and $x_j <x_i$ in $\Or_2$, where $i,j \in [1,m]$.
Moreover, it must be the case that either $x_i < x_j$ or $x_j <x_i$ is the correct order in $\llist$, since
both $x_i, x_j \in \llist$.
(Note that, due to the security of bilinear aggregate signature scheme, it must be the case that all the elements of $\sublist$ are indeed 
elements of $\llist$, i.e, $ x_1, x_2, \ldots, x_m \in \llist$ (except with negligible probability).)

Without loss of generality, let us assume $x_i < x_j$ is the order in $\llist$ and $\mathsf{rank}(\llist,x_i) = u < v= \mathsf{rank}(\llist,x_j)$. This implies $x_j <x_i$ is the forged order for which $\Q'$ has successfully 
generated a valid proof, i.e., $e(t_{x_j \in \llist}, t_{x_j < x_i}) = e(t_{x_i \in \llist}, g)$ has verified since $\ve$ accepted the corresponding proof.
We show that by invoking $Q'$ and using its output, $t_{x_j < x_i}$ , we construct a PPT adversary $\advM$ that successfully 
solves the \emph{n-BDHI Problem} \cite{Boneh04efficientselective-id} thereby contradicting $n$-Bilinear Diffie Hellman assumption.
The formal reduction follows:

\begin{Theorem}
\label{thm:SO}
If $n$-Bilinear Diffie Hellman assumption holds, then $\ppads$ scheme satisfies Soundness in Definition~\ref{def:soundness}.
\end{Theorem}

\begin{proof}
 We show that if there exists a malicious $\Q'$ as discussed above, then we construct a PPT adversary $\advM$ that successfully 
solves the \emph{n-BDHI Problem} \cite{Boneh04efficientselective-id}. Algorithm $\advM$ is given the public parameters $(p, G, G_T, e, g)$ and
$\mathcal{T} = \langle g, g^s,g^{s^2}, \ldots ,g^{s^n} \rangle$, where $n = \p$. $\advM$ runs as follows:

\begin{enumerate}[noitemsep,nolistsep]

     \item Pick $ \x \xleftarrow{\$} \ZZ_p^\ast$ a list $\llist$ such that $|\llist| = n$.\\
     Pick $\R = \lbrace{r_i \xleftarrow{\$} \ZZ_p^\ast \rbrace}_{\forall i \in [1,n]}$ and compute $t_{x_i \in \llist} \leftarrow {(g^{s_i})}^{r_i}\; \forall i \in [1,n]$.\\
     Compute ${\sigma_i} \leftarrow {\mathcal{H}(t_{x_i \in \llist}||x_i)}^\x $, $\forall x_i \in \llist$.\\
     Pick the $\mathsf{nonce}, \; {\omega} \xleftarrow{\$} {\lbrace 0,1 \rbrace}^{\ast}$ and compute $\salt \leftarrow {({\mathcal{H}}({{\omega}})})^\x$.\\
     The list digest signature is computed as: $ {\sigma}_\llist \leftarrow  \salt \times {\prod}_{1 \leq i \leq n} {\sigma_i}$.\\
     Set $\K = \{g^{\x},\sigma_\llist, \mathcal{T},\Sigma_\llist, \R\}$ where $\Sigma_\llist = \langle \lbrace  {{t_{x_i \in \llist}, \sigma_i}} {\rbrace}_{1 \leq i \leq n}, \mathcal{H}({\omega}) \rangle$.

     \item Finally $\Q'$ outputs two contradicting orders $\Or_1 \neq \Or_2$ for some sublist, $\sublist =\lbrace x_1, x_2, \ldots, x_m \rbrace$.\\
     As discussed above, let $(x_i,x_j)$ be an inversion pair such that $x_i < x_j$ is the order in $\llist$ and $\mathsf{rank}(\llist,x_i) = u < v= \mathsf{rank}(\llist,x_i)$.\\ 
     This implies $x_j <x_i$ is the forged order for which $\Q'$ has successfully generated a valid proof $t_{x_j < x_i}  = {(g^{s^{(u-v)}})}^{r_2 r_1^{-1}}$.
     
     \item Now $\advM$ outputs $e($\authunito{x_j}{x_i}$, {(g^{s^{v-u-1}})}^{{r_2}^{-1} r_1}) = {e(g,g)}^{\frac{1}{s}}$.

\end{enumerate}

$\advM$ inherits success probability of $\Q'$, therefore if $\Q'$ succeeds with non-negligible advantage, so does $\advM$. Hence, a contradiction. \qed

\end{proof}

\subsection{Proof of Zero-Knowledge}

We define Zero Knowledge Simulator $\Sim = (\Sim_1, \Sim_2)$ from Definition~\ref{def:ZK} as follows.
$\Sim$ has access to the system parameters $(p,G,G_1,e,g,\mathcal{H})$ and executes the following steps:
\begin{itemize}[noitemsep,nolistsep]
 \item $\Sim_1$ picks a random element $\x\xleftarrow{\$} \ZZ_p^{\ast}$ and a random element $g_1 \xleftarrow{\$} G$ and publishes as $\D = (g^v, g_1^v)$ and keeps $\x$ as the secret key.
 \item $\Sim_2$ maintains a table of the elements already queried of tuples $\langle x_i, r_i \rangle$ where $x_i$ is the element already queried and $r_i$ is the corresponding random element picked from $\ZZ_p^\ast$ by $\Sim_2$.

 For a query on sublist $\sublist = \{ x_1, x_2, \ldots, x_m \}$, $\Sim_2$ makes an oracle access to list $\llist$ to get the list order of the elements. Let us call it $\Or = \pi_{\llist}(\sublist) = \{y_1, y_2, \ldots, y_m\}$.
 \begin{itemize}[noitemsep,nolistsep]
 \item For every $i \in [1,m]$ $\Sim_2$ checks if $y_i$ is in the table. If it is,
$\Sim_2$ uses the corresponding random element from the table.
 Otherwise, $\Sim_2$ picks a random element 
 $r_i \xleftarrow{\$} \ZZ_p^{\ast}$ and adds $\langle y_i, r_i \rangle$ to the table.
 \item $\Sim_2$ sets the member authentication unit as $t_{y_i \in \llist} := g^{r_i}$ and computes $\sigma_i \leftarrow \mathcal{H}{(t_{y_i \in \llist} || y_i)}^\x$.
 \item $\Sim_2$ sets $\sigma_\Or := {\prod}_{y_i \in \Or} {\sigma_i}$ and $\lambda_{\llist'} := \frac{g_1}{{\prod}_{{y_i \in \Or}} {\mathcal{H}(t_{y_i\in \llist}||y_i)}}\ .$
 \item For every pair of elements $y_i, y_{i+1}$ in $\Or$, $\Sim_2$ computes $t_{y_i <y_{i+1}} \leftarrow g^{r_{i+1} /r_i}\ .$
 \item Finally, $\Sim_2$ returns $\Or,\pr$, where $\pr = ( \Sigma_\Or,\Omega_\Or)$, $\Sigma_\Or = ( \sigma_\Or,T, \lambda_{\llist'})$, $T = \lbrace  t_{y_1 \in \llist}, \ldots,  t_{y_m \in \llist} \rbrace$
 and  $\Omega_\Or =  \{$\authunito{y_1}{y_{2}}, \authunito{y_2}{y_{3}}, \ldots,\authunito{y_{m-1}}{y_{m}}$\}$.
\end{itemize} 
\end{itemize}

The simulator $\Sim = (\Sim_1, \Sim_2)$ produces outputs that are identically distributed to the distribution outputs of the true $\setupM$ and $\Q$ algorithms.
In both cases $v$ is picked randomly.
Let $x,y,z \in \ZZ_p^{\ast}$ where $x$ is a fixed element and $z = xy$. Then $z$ is identically distributed to $y$ in $\ZZ_p^{\ast}$.
In other words, if $y$ is picked with probability $\gamma$, then so is $z$.
The same argument holds for elements in~$G$ and $G_1$.
Therefore all the units of $\Sigma_\Or$ and $\Omega_\Or$ are distributed identically in both cases. Thus our PPAL scheme is simulatable and the Zero-Knowledge is perfect. \qed

We summarize the properties and efficiency of our PPAL construction in Theorem~\ref{thm:PPAL}.

\begin{Theorem}
\label{thm:PPAL}

The privacy-preserving authenticated list (PPAL) construction of
Figure~\ref{fig:protocol} satisfies the security properties of
completeness (Definition~\ref{def:completeness}), soundness
(Definition~\ref{def:soundness}) and zero-knowledge
(Definition~\ref{def:ZK}).  Also, the construction has the following
performance, where $n$ denotes the list size and $m$ denotes the query
size.
\begin{itemize} \itemsep 0pt
\item The owner and server use $O(n)$ space.
\item The owner performs the setup phase in $O(n)$ time.
\item  The server performs the preprocessing phase in~$O(n)$ time.
\item The server computes the answer to a query and its proof in
  $O(\min \lbrace m\log n,n \rbrace)$ time.
\item  The client verifies the proof in $O(m)$ time and space.
\end{itemize}

\end{Theorem}

\section{Acknowledgment}
\label{sec:ack}
This research was supported in part by
the National Science Foundation under grant CNS--1228485.
Olga Ohrimenko worked on this project in part while at Brown
University.
We are grateful to Melissa Chase, Anna Lysyanskaya, Markulf
Kohlweiss, and Claire Mathieu for useful discussions and for their
feedback on early drafts of this work.
We would also like to thank Ashish Kundu for introducing us to his
work on structural signatures and Jia Xu for sharing a paper through
personal communication.


\bibliographystyle{alpha}

\bibliography{proposal}

\clearpage

\appendix
\label{app}

\noindent{\Large \textbf{Appendix}}\\

\section{Homomorphic Integer Commitment Scheme~\cite{DamgardF02} and its Simulator}
\label{sec:IC}

We write the commitment scheme of~\cite{DamgardF02}, in the trusted parameter model, i.e., the public key is generated by a trusted third party.
However, in the original paper~\cite{DamgardF02}, the prover and the verifier interactively set up the public parameters. 

\begin{figure}[H]
\caption{ Homomorphic Integer Commitment Scheme \cite{DamgardF02}.}
\begin{framed}
\begin{description}[noitemsep,nolistsep]
\item $\HomIntCom = (\cS,$ $\IC,$ $\IO)$
\begin{description}[nolistsep,noitemsep]
 \item [$\PK_C \leftarrow \cS(1^k)$:] The $\cS$ algorithm, takes the security parameter as input
 and generates the description of a finite Abelian group $\GG$, $desc(\GG)$, and a large integer $F(k)$ such that 
it is feasible to factor numbers that are smaller than $F(k)$. A number having only prime factors at most $F(k)$ are called $F(k)$-smooth and a number having prime factors larger then $F(k)$ are called $F(k)$-rough.
 The algorithm then chooses a random element $h \xleftarrow{\$} \GG$ (by group assumption, $ord(h)$ is $F(k)$-rough with overwhelming probability) and a random secret key $s \xleftarrow{\$} \ZZ_{ord(\GG)}$
 and sets $g := h^s$.~$\cS$ outputs $(desc(\GG),F(k),g,h)$ as the public key of the commitment scheme, $\PK_C$.
 
 \item [$(c,r) \leftarrow \IC(\PK_C,x)$:] To commit to an integer $x$, the algorithm~$\IC$ chooses a random $r$, $r \xleftarrow{\$} \ZZ_{2^{B+k}}$, and computes 
 $c= g^x h^r$ (where $B$ is a reasonably close upper bound on the order of the group $\GG$, i.e., $2^B > ord(\GG)$, and given $desc(\GG)$, $B$ can be computed efficiently).
 $\IC$ outputs $(c,r)$.
 
 \item [$x \leftarrow \IO(\PK_C,c,r)$:] To open a commitment $c$, the committer must send the opening information $(x,r,b)$ to the verifier 
 such that $c = g^x h^r b$ and $b^2 = 1$. An honest committer can always set $b := 1$.
\end{description}
\end{description}
\end{framed}
\label{fig:IC}
\end{figure}

The above commitment scheme is \emph{homomorphic} as

$$\IC(\PK_C,x + y) =  \IC(\PK_C,x) \times \IC(\PK_C,y).$$

In Figure~\ref{fig:SimIC} we present a simulator for $\HomIntCom$.
We note that the distribution of outputs from the simulator algorithms is identical to the distribution of outputs from a true prover (committer): 
in both cases $desc(\GG),F(k),g,h$ and commitments are generated identically.

\paragraph{Efficiency}
Assuming group exponentiation take constant time, both $\IC$ and $\IO$ run in asymptotic time $O(1)$.

\begin{figure}[H]
\caption{Simulator for $\HomIntCom$.}
\begin{framed}

\begin{description}[nolistsep,noitemsep]
\item $\Sim\HomIntCom = (\Sim\cS,$ $\Sim\IC,$ $\Sim\IO)$

\begin{description}[nolistsep,noitemsep]
 \item [$(\PK_C, \TK_C) \leftarrow \Sim\cS(1^k)$:] $\Sim\cS$ works exactly as the $\cS$ except that it saves $s$ and the order of the group~$\GG$, $ord(\GG)$.
 $\sICS$ sets $\TK_C = (ord(\GG),s)$ and outputs $(\PK_C= (desc(\GG),F(k),g,h),\TK_C)$.
 
 \item [$(c,r) \leftarrow \Sim\IC(\PK_C,x)$:] $\Sim\IC$ behaves exactly as $\IC$ and outputs $(c,r)$
 where $c = g^x h^r$, $r \xleftarrow{\$} \ZZ_{2^{B+k}}$  and $B$ is as defined in Figure~\ref{fig:IC}.

 \item [$x' \leftarrow \Sim\IO(\PK_C,\TK_C,c,r)$:] To open  a commitment $c$, originally committed to some integer $x$,
 to any arbitrary integer $x' \neq x$, send $(x', (r+ sx - sx')\mod ord(\GG), b=1)$ to the verifier.
\end{description}
\end{description}

\end{framed}
\label{fig:SimIC}
\end{figure}



\newpage
\section{Proving an Integer is Non-negative~\cite{Lipmaa03}}

We present the $\Sigma$ protocol presented in~\cite{Lipmaa03} in Figure~\ref{fig:Lipmaa}. This protocol is honest-verifier zero knowledge with 3 rounds of interaction
and can be converted to non-interactive general zero knowledge in the Random Oracle model using Fiat-Shamir heuristic \cite{FiatS86}.

The protocol is essentially based on two facts: a negative number cannot be a sum
of squares and every non-negative integer is a sum of four squared integers. The representation
of a non-negative integer as the sum of four squares is called the Lagrange representation of a non-negative integer.
\cite{Lipmaa03} presents an efficient probabilistic time algorithm to compute the Lagrange representation of a non-negative integer.

\begin{Theorem}\cite{Lipmaa03}
An integer $x$ can be represented as $ x= \omega_1^2 + \omega_2^2 + \omega_3^2 + \omega_4^2$, with integer $\omega_i$, $i \in [1,4]$, iff $x \geq 0$.
Moreover, if $x \geq 0$, then the corresponding representation $(\omega_1,\omega_2,\omega_3,\omega_4)$ can be computed efficiently.
\end{Theorem}

\paragraph{Efficiency}
The algorithm to compute Lagrange's representation of a non-negative integer is probabilistic polynomial time \cite{Lipmaa03}.
Assuming group exponentiation is done in constant time, both the \emph{Prover} and the \emph{Verifier} in the protocol in Figure~\ref{fig:Lipmaa}
run in asymptotic constant time, i.e., $O(1)$.

\begin{figure}[H]
\caption{Proving non-negativity of an integer \cite{Lipmaa03}: $\Pro(x,r :c = C(x;r) \wedge x \geq 0)$}
\begin{framed}

\begin{description}[nolistsep,noitemsep]
 \item [Step 1:] The \emph{Prover} commits to an integer $x \in \{-M,M\}$ as $c := \IC(\PK_C,x) = g^xh^\rho$ where $\rho \in \ZZ_{2^{B+k}}$ and sends it to the \emph{Verifier}.
 Now the \emph{Prover} computes the following:
 \begin{itemize}[nolistsep,noitemsep]
  \item represent $x$ as $x = \omega_1^2 + \omega_2^2 + \omega_3^2 + \omega_4^2$
  \item for $i \in [1,4]$: pick $r_{1i} \xleftarrow{\$} \ZZ_{2^{B+2k}}$ such that $\sum_i r_{1i} = \rho$
  \item for $i \in [1,4]$: pick $r_{2i} \xleftarrow{\$} \ZZ_{2^{B+2k}F(k)}$ and $r_3 \xleftarrow{\$} \ZZ_{2^{B+2k}F(k)\sqrt{M}}$
  \item for $i \in [1,4]$: pick $m_{1i} \xleftarrow{\$} \ZZ_{2^{k}F(k)\sqrt{M}}$
  \item for $i \in [1,4]$: compute $c_{1i} \leftarrow g^{\omega_1}h^{r_{1i}}$
  \item compute $c_2 \leftarrow g^{\sum_i m_{1i}}h^{\sum_i r_{1i}}$
  \item compute $c_3 \leftarrow (\prod_i {c_{1i}}^{m_{1i}}) h^{r_3}$
 \end{itemize}
 The \emph{Prover} sends $(c_{11},c_{12},c_{13},c_{14},c_2,c_3)$ to the \emph{Verifier}.
 
 \item[Step 2:] The \emph{Verifier} generates $e \xleftarrow{\$} \ZZ_{F(k)}$ and sends it to the \emph{Prover}.
 
 \item[Step 3:] The \emph{Prover} computes the following:
 \begin{itemize}[nolistsep,noitemsep]
  \item for $i \in [1,4]$: compute $m_{2i} \leftarrow m_{1i} + e{\omega_i}$
  \item for $i \in [1,4]$: compute $r_{4i} \leftarrow r_{2i} + er_{1i}$
  \item compute $r_5 \leftarrow  r_3 + e\sum_i(1-\omega_i)r_{1i}$
 \end{itemize}
 The \emph{Prover} sends $(m_{21},m_{22}, m_{23},m_{24},r_{41},r_{42},r_{43},r_{44},r_5)$ to the \emph{Verifier}.
 
 \item[Step 4:] The \emph{Verifier} checks the following:
 \begin{itemize}[nolistsep,noitemsep]
  \item for $i \in [1,4]$: check $g^{m_{2i}}h^{r_{4i}}{c_{1i}}^{-e} \stackrel{?}{=} c_2$
  \item $(\prod_i {c_{1i}}^{m_{2i}}) h^{r_5}c^{-e} \stackrel{?}{=} c_3$
 \end{itemize}
 
\end{description}

\end{framed}
\label{fig:Lipmaa}
\end{figure}


\newpage
\section{Zero Knowledge Set (ZKS) Construction~\cite{ChaseHLMR05}}
\label{sec:ZKSCons}
Here we give the construction of ZKS based on mercurial commitments and collision-resistant hash
functions. For the details, please refer to Section 3 of \cite{ChaseHLMR05}.

For a finite database $D$, the prover views each key $x$ as an integer numbering of a leaf of a height-$l$ binary tree and places a commitment to the information $v = D(x)$ into leaf number $x$.
To generate the commitment $C_D$ to the database $D$, the prover~$\PD$ generates an incomplete binary tree of commitments, resembling a Merkle tree as follows.
Let $\mathsf{Merc} = \{\MS,\HC,$ $\SC,$ $\tease,$ $\VT,\open,$ $\VO\}$ be a Mercurial Commitment scheme and~$\PK_D$ be the public key of the mercurial commitment scheme, i.e., $\PK_D \leftarrow \MS(1^k)$.
Let~$r_x$ denotes the randomness used to produce the commitment (hard or soft) of $x$.

Before getting into the details of the ZKS construction using mercurial commitments in Figure~\ref{fig:ZKS},
let us give an informal description of mercurial commitments.
Mercurial commitments slightly relax the binding property of commitments. 
Partial opening, which is called \emph{teasing}, is not truly binding: it is possible for the
committer to come up with a commitment that can be teased to any value of its
choice. True opening, on the other hand, is binding in the traditional
sense: it is infeasible for the committer to come up with a commitment that it can open to
two different values. If the committer can open a commitment at all,
then it can be teased to only one value. Thus, the committer must choose, at the time of
commitment, whether to \emph{soft-commit}, so as to be able to tease to multiple values but
not open at all, or to \emph{hard-commit}, so as to be able to tease and to open to only one
particular value. It is important to note that hard-commitments and soft-commitments are computationally indistinguishable.

\paragraph{Efficiency}

Let us assume that the elements are hashed to $k$ bit strings, so that $l=k$. Let us also assume (as in \cite{ChaseHLMR05}) that 
the collision resistant hash is built into the mercurial commitment scheme, allowing
to form $k$-bit commitments to pairs of $k$-bit strings.
Therefore, computing the commitment $\com$ takes time $O(ln) = O(kn)$, where $|D| = n$.

The proofs of membership and non-membership consists of $O(k)$ mercurial decommitments each
and the verifier needs to verify $O(l) = O(k)$ mercurial decommitments to accept the proof's validity.

A constant time speed-up can be achieved using the $q$-Trapdoor Mercurial Commitment ($q$-TMC) scheme and collision resistant hash function as building blocks.
$q$-TMC was introduced by \cite{Catalano:2008} and later improved by \cite{Libert2010}.
The construction is similar to ~\cite{ChaseHLMR05}, except a $q$-ary tree of height $h$ is used ($q$ >2) instead of a binary tree and each leaf is expressed in $q$-ary encoding.
Using $q$-TMC as a building block achieves significant improvement in ZKS implementation \cite{Catalano:2008,Libert2010} though the improvement is not asymptotic.


\begin{figure}[H]
\caption{Zero Knowledge Set (ZKS) construction from Mercurial Commitments \cite{ChaseHLMR05}.}
\begin{framed}
\begin{description}[nolistsep,noitemsep]
 \item $\zks = (\zks\setupM,\zks\ZP = (\zks \Pone, \zks \Ptwo),\zks\ZV)$

\begin{description}[noitemsep,nolistsep]
 \item[$\PK_D \leftarrow \zks\setupM(1^k)$:] Run $\PK_D \leftarrow \MS(1^k)$ and output $\PK_D$.
 
 \item[$(\com,\st) \leftarrow \zks \Pone(1^k, \PK_D, D)$:] $\zks \Pone$ runs as follows:
 \begin{itemize}[nolistsep,noitemsep]
 \item For each $x$ such that $D(x) = v \neq \bot$, produce $C_x = \HC(\PK_D,v,r_x)$. 
 \item For each $x$ such that $D(x) = \bot$ but $D(x') \neq \bot$, where $x'$ denotes $x$ with the last bit flipped, produce $C_x = \SC(\PK_D,r_x)$.
 \item Define $C_x = \n$ for all other $x$ and build the tree in bottom up fashion. For each level $i$ from $l-1$ upto 0, and for each string $\sigma$ of length $i$, define the commitment $C_\sigma$ as follows: 
 \begin{enumerate}[nolistsep,noitemsep]
  \item For all $\sigma$ such that $C_{\sigma0} \neq \n \wedge C_{\sigma1} \neq \n$, let $C_\sigma = \HC(\PK_D,(C_{\sigma0}, C_{\sigma1}), r_\sigma)$.
  \item For all $\sigma$ such that $C_{\sigma'}$ have been defined in Step 1 (where $\sigma'$ denotes $\sigma$ with the last bit flipped) but $C_\sigma$ has not, define $C_\sigma = \SC(\PK_D,r_\sigma)$.
  \item For all other $\sigma$, define $C_\sigma = \n$.
 \end{enumerate}
 \item If the value of the root, $C_\epsilon = \n$, redefine $C_\epsilon = \SC(\PK_D,r_\epsilon)$. This happens only when $D = \phi$. 
 Finally define $C_D = C_\epsilon = \com$.

 \end{itemize}

 \item[$(D(x),\Pi_x) \leftarrow \zks \Ptwo(\PK_D,\st, x)$:] For a query $x$, $\zks \Ptwo$ runs as follows:
 
 \begin{description}[nolistsep,noitemsep]
  \item [$\mathbf{x\in D}$, i.e., $D(x) = v \neq \bot$:]
  Let $(x\lvert i)$ denote the first $i$ bits of the string $x$ and $(x\lvert i)'$ be $(x\lvert i)$ with the last bit flipped.
  Let $\pr_x = \open(\PK_D,D(x),r_x,C_x)$ and $\pr_{(x\lvert i)} = \open(PK_D,(C_{(x\lvert i 0)},C_{(x\lvert i 1)}),r_{(x\lvert i)},C_{(x\lvert i)})$ for 
  all $0 \leq i < l$, where $C_{(x\lvert i)}$ is a commitment to its two children $C_{(x\lvert i 0)}$ and $C_{(x\lvert i 1)}$.
  
  Return $(D(x), \Pi_x = ({\{C_{(x\lvert i)},C_{(x\lvert i)'}\}}_{i \in [1,l]}, {\{\pr_{(x\lvert i)}\}}_{i \in [0,l]}))$.  
  
  \item [$\mathbf{x\not\in D}$, i.e., $D(x) = \bot$:]
  If $C_x = \n$, let $h$ be the largest value such that $C_{(x\lvert h)} \neq \n$, let $C_x = \HC(\PK_D,\bot,r_x)$ and build a path from $x$ to $C_{(x\lvert h)}$ as follows:
  define $C_{(x\lvert i)} = \HC(\PK_D,(C_{(x\lvert i 0)},C_{(x\lvert i 1)}),r_{(x\lvert i)}), C_{(x\lvert i)'} = \SC(\PK_D,r_{(x\lvert i)'})$ for all $i \in [l-1,h+1]$.
  Let $\tau_x = \tease(\PK_D, D(x),r_x,C_x)$ and
  $\tau_{(x\lvert i)} = \tease(PK_D,(C_{(x\lvert i 0)},C_{(x\lvert i 1)}),r_{(x\lvert i)},C_{(x\lvert i)})$ for all $ 0 \leq i < l$.
  
   Return $(\bot, \Pi_x=( {\{C_{(x\lvert i)},C_{(x\lvert i)'}\}}_{i \in [1,l]}, {\{\tau_{(x\lvert i)}\}}_{i \in [0,l]}))$
 \end{description}

 \item[$b \leftarrow \zks \ZV(1^k, \PK_D, \com,x,D(x),\Pi_x)$:]~
 
 \begin{description}[nolistsep,noitemsep]
  \item [$\mathbf{x \neq \bot}$:] The verifier $\zks\ZV$ performs the following:
  \begin{itemize} [nolistsep,noitemsep]
   \item $\VO(\PK_D, C_{(x\lvert i)},(C_{(x\lvert i 0)},C_{(x\lvert i 1)}),\pr_x)$ for all $1 \leq i <l$
   \item $\VO(\PK_D,C_D,(C_0,C_1),\pr_\epsilon)$ and $\VO(\PK_D,C_x,D(x),\pr_x)$.
  \end{itemize}
  
  \item [$\mathbf{x = \bot}$:] The verifier $\VD$ performs the following:
  \begin{itemize} [nolistsep,noitemsep]
   \item $\VT(\PK_D, C_{(x\lvert i)},(C_{(x\lvert i 0)},C_{(x\lvert i 1)}),\tau_x)$ for all $1 \leq i < l$
   \item $\VT(\PK_D,C_D,(C_0,C_1),\tau_\epsilon)$ and $\VT(\PK_D,C_x,\bot,\tau_x)$
  \end{itemize}

 \end{description}
 
 \end{description}
 \end{description}

\end{framed}
\label{fig:ZKS}
\end{figure}


\end{document}